\documentclass[12pt, a4paper, fleqn]{article}

\usepackage{verbatim}
\usepackage{etex}
\usepackage[utf8]{inputenc}
\usepackage[TS1,T1]{fontenc}
\usepackage{xspace}
\usepackage{array}
\usepackage{tikz-cd}
\usetikzlibrary{external}
\usepackage{smartdiagram}
\usepackage{tensor}
\usetikzlibrary{decorations.markings}
\usepackage{layout} 
\usetikzlibrary{arrows.meta} 
\usepackage{pdfpages}
\usepackage[a4paper,tmargin=2.75truecm,bmargin=2.5truecm,hmargin=2.8truecm,  rmargin=2.9truecm, lmargin=2.9truecm, marginparwidth=25mm]{geometry}

\usepackage[english]{babel}
\usepackage{blindtext} 
\usepackage{epigraph} 
\usepackage{amsfonts, mathtools,slashed,bm}
\usepackage[bb=ams]{mathalpha}
\usepackage{slashed}
\usepackage{enumitem}
\newlist{enum-hypothesis}{enumerate}{1}
\setlist[enum-hypothesis]{label=(\arabic*),itemsep=0pt, parsep=0pt}

\setlist[enumerate,1]{label=\arabic*., ref=\arabic*, topsep=1pt, itemsep=2pt, parsep=0pt, leftmargin=1.5em, itemindent=0em, labelsep=0.2em, labelwidth=1.3em}
\setlist[enumerate,2]{label=\alph*., ref=\theenumi.\alph*, topsep=1pt, itemsep=2pt, parsep=0pt, leftmargin=0.5em, itemindent=0em, labelsep=0.2em, labelwidth=1.5em}
\setlist[enumerate,3]{label=\roman*., ref=\theenumii.\roman*, topsep=1pt, itemsep=2pt, parsep=0pt, leftmargin=0.5em, itemindent=0em, labelsep=0.2em, labelwidth=1.2em}

\usepackage{booktabs, dcolumn}
\usepackage{graphicx, subfig}

\usepackage{pbox}

\usepackage[thmmarks,amsmath]{ntheorem}

\newtheorem{theorem}{Theorem}[section]
\newtheorem{proposition}[theorem]{Proposition}
\newtheorem{lemma}[theorem]{Lemma}
\newtheorem{corollary}[theorem]{Corollary}

\newtheorem{definition}[theorem]{Definition}

\theoremsymbol{$\square$}
\theoremstyle{plain}
\theorembodyfont{\upshape}
\newtheorem{remark}[theorem]{Remark}
\theoremsymbol{$\Diamond$}

\theoremsymbol{}
\theoremstyle{break}
\theorembodyfont{\slshape}

\theoremheaderfont{\scshape}
\theorembodyfont{\upshape} 
\theoremstyle{nonumberplain}
\theoremseparator{} 
\theoremsymbol{\rule{1.3ex}{1.3ex}} 
\newtheorem{proof}{Proof}

\usepackage{hyperref}

\hypersetup{
	plainpages=false,
	colorlinks=true,
	linkcolor=blue, 
	anchorcolor=black, 
	citecolor=blue, 
	urlcolor=blue, 
	menucolor=black, 
	filecolor=black, 
	bookmarksopen=true,
	bookmarksnumbered=true}

\hypersetup{
	pdfauthor={Gaston Nieuviarts},
	pdftitle={Titre},
	pdfsubject={},
	pdfcreator={pdflatex},
	pdfproducer={pdflatex},
	pdfkeywords={}
}

\newcommand{\Man}{\mathcal{M}}
 
\newcommand{\bbbone}{{\mathbb I}} 

\newcommand{\ds}{\slashed\partial}
\newcommand{\calA}{\mathcal{A}}
\newcommand{\calB}{\mathcal{B}}

\newcommand{\calH}{\mathcal{H}}
\newcommand{\calJ}{\mathcal{J}}

\newcommand{\calS}{\mathcal{S}}

\newcommand{\calU}{\mathcal{U}}

\newcommand{\J}{\mathcal{J}}

\newcommand{\hs}{\calH}

\newcommand{\ura}{\mathcal{U}_\rho({\mathcal A})}

\newcommand{\Ad}{\text{Ad}}

\newcommand{\defeq}{\vcentcolon=} 

\DeclareMathOperator{\Tr}{Tr}	   

\newcommand{\Ths}[1][]{\widetilde{\hs}}

\newcommand{\act}{\calS} 

\newcommand{\Sp}{\mathcal{S}} 
\newcommand{\Dir}{D} 
\newcommand{\Kont}{K} 

\newcommand{\A}{{\cal A}} 
\newcommand{\HH}{{\cal H}} 
\newcommand{\I}{{\mathbb I}} 
\newcommand{\cinf}{{C^\infty(\Man)}} 
\newcommand{\C}{\mathbb C}
\newcommand{\OO}{{\mathcal O}}

\newcounter{mnotecount}[section]
\renewcommand{\themnotecount}{\thesection.\arabic{mnotecount}}
\newcommand{\mnote}[1]%
{\protect{\stepcounter{mnotecount}}${}^{\text{\footnotesize$\bullet$\themnotecount}}$%
	\reversemarginpar%
	\marginpar{\raggedleft\footnotesize$\bullet$\themnotecount: #1}}

\newlength{\mnotewidth}
\definecolor{blueamu}{RGB}{0, 101, 189}
\definecolor{cyanamu}{RGB}{61, 183, 228}
\newcommand{\dhorline}[3][0]{%
	\tikz[baseline=-2pt]{\path[decoration={markings, 
			mark=between positions 0 and 1 step 2*#3
			with {\node[color=blueamu, fill, circle, minimum width=#3, inner sep=0pt, anchor=south west] {};}},postaction={decorate}]  (0,#1) -- ++(#2,0);}}
\newcommand{\dvertline}[3][0]{%
	\tikz[baseline=2em]{\path[decoration={markings,
			mark=between positions 0 and 1 step 2*#2
			with {\node[color=black!50, fill, circle, minimum width=#2, inner sep=0pt, anchor=south west] {};}},postaction={decorate}] (0, #1) -- ++(0,#3);}}

\makeatletter\newcommand\HUGE{\@setfontsize\Huge{28}{0}}\makeatother		
\numberwithin{equation}{section}

\allowdisplaybreaks
\begin{document}
	\renewcommand\figurename{Fig.}
	
	{
		\makeatletter\def\@fnsymbol{\@arabic}\makeatother 
		\title{\vspace{-1.5truecm}	 Torsion and Lorentz symmetry \\from   twisted
                  spectral triples \vspace{-.0truecm}}
		\author{$^a$P. Martinetti, $^a$G. Nieuviarts, $^b$R. Zeitoun\footnote{martinetti@dima.unige.it, gaston.nieuviarts@etu.univ-amu.fr, ruben.zeitoun@ens-lyon.fr}\\
			{\small $^a$DIMA, Universita di Genova \& INFN sezione di Genova}\\
			{\small $^b$Laboratoire AGM, CY Cergy Paris Université}\\
			%
			\\
		}
	\date{}
	\maketitle
	}  
	
	\setcounter{tocdepth}{3}
	
		\vspace{-1.75truecm}
	\begin{abstract}
By twisting the spectral triple of a riemannian spin
manifold, we show how to generate an orthogonal and geodesic
preserving torsion from a torsionless Dirac operator. We identify the group of
                twisted unitaries as the generator of torsion with
                co-exact three form. 
                Through the fermionic action, the
                torsion term identifies with a Lorentzian energy-momentum
                $4$-vector. The 
                Lorentz group turns out to be a normal subgroup of the
                twisted unitaries. We also
                investigate the spectral action related to this model.
	\end{abstract}
\small{	\tableofcontents }

	\section{Introduction}
	\label{sec introduction}
	
	One of the achievements of the spectral description of  the Standard Model of particles physics \cite{chamseddine2007gravity}  is to obtain
	the Higgs field on the same footing as the other gauge bosons, that is a
	component of a connection $1$-form. The latter is a suitable
        generalisation  to the noncommutative setting of 
	the Levi-Civita connection. 

	The appearance of Levi-Civita connection (the unique
	torsionless connection on the tangent bundle compatible
	with the metric) is
	not surprising.  Its use is customary in general
	relativity, where  the compatibility with the metric guarantees that the pseudo-norm of
	timelike vectors is preserved by parallel transport. However there
	does not seem to be strong theoretical motivations to impose vanishing
	torsion. 
Actually the interplay between torsion and relativity is
	an old and long story, from the early work of Cartan 
	to recent
	applications in neutrinos oscillations
	\cite{TorNeutriOsc, alimohammadi1999neutrino}, parity
	violation \cite{sengupta1999parity,freidel2005quantum,mukhopadhyaya2002space}, cosmology
	\cite{tilquin2011torsion, poplawski2011spacetime,
          minkevich2007gravitation} or the problem of singularities
        \cite{trautman1973spin,luz2020singularity}. Good reviews are
	\cite{shapiro2002physical} and \cite{hehl1976general}.
	
	In noncommutative geometry as well, torsion has been
        investigated. The spectral action for a Dirac operator with a
        certain kind of torsion has been
	computed in
	\cite{SpectActTorsStephan2010,pfaffle2011holst,iochum2012spectral},
	and more formal developments have been recently proposed in \cite{dkabrowski2023spectral}. 

	In this paper we explore an alternative way, consisting in generating
	torsion from a torsionless connection, through 
	a \emph{twisted fluctuation} of the Dirac operator.
      In addition, this torsion turns out to
        be related to a change of signature from the euclidean to the lorentzian.

More precisely, let us recall that in the spectral description of the
Standard Model \cite{chamseddine2007gravity}, all the bosonic fields are obtained by fluctuating
(definition is recalled in \S \ref{subsec:twistedfluc}) the generalized
	Dirac operator~\cite{Connes:1996fu}
	\begin{equation}
		\label{eq:Dirac}
		\slashed\partial\otimes \mathbb {\mathbb I}_F + \gamma\otimes D_F
	\end{equation}
	where $\mathbb I_F$ is the identity on
	the finite dimensional Hilbert space  $\mathcal H_F$ spanned by fermions, $D_F$ is a matrix on $\mathcal H_F$ that contains the
	parameters of the model (Yukawa coupling of fermions and
	mixing
	angles for quarks and neutrinos), while $\gamma$ is a $\mathbb
        Z^2$-grading{\footnote{This is  the generalisation to arbitrary even dimension of the fifth Dirac
        matrix $\gamma^5$, see appendix \ref{ConvGamMat} for
        details.}} of
	the Hilbert space $L^2(\Man, S)$ of square integrable spinors on a
	closed, orientable, riemannian, spin 
	manifold $\Man$ of even dimension $n=2m$ and  
	\begin{equation}
		\label{eq:freeDirac}
		\slashed\partial=-i\gamma^\mu\widetilde\nabla_\mu^S
	\end{equation}
	is the usual
Dirac operator. In the formula above,  $\gamma^\mu$ for $\mu=1, ...,
n$  are the Dirac matrices  
	while 
	$\widetilde\nabla_\mu^S$ is the spin connection, that is the lift from the tangent to the 
     spinor bundle of the Levi-Civita connection.

	Twisted fluctuations have been
	introduced for the spectral triple of Standard Model \cite{TwistSpontBreakDevastaMartine2017,filaci2021minimal,filaci2023critical} with the
	aim of generating an extra scalar field required to fit the
	Higgs mass and stabilise the electroweak vacuum
	\cite{chamseddine2012resilience}. This extra scalar field is obtained
	from the component $\gamma\otimes D_F$ of the
	operator \eqref{eq:Dirac}. However in the process also 
       the free part 
	$\slashed\partial\otimes\mathbb I_F$ twist-fluctuates and generates an
        unexpected field of
	$1$-forms $f_\mu
        dx^\mu$ \cite{Devastato:2013fk,TwistSpontBreakDevastaMartine2017}.
       
The interpretation of the field was unclear so far, except in one example: the twist of the spectral triple of
	electrodynamics \cite{Dungen:2011fk}.   By computing the 
	fermionic action \cite{martinetti2022lorentzian,devastato2018lorentz},
	one gets that $f_\mu dx^\mu$ 
	identifies with the
	(dual of) the energy-momentum $4$-vector in lorentzian
        signature, although one starts with a riemannian
        manifold. This is in line with previous results pointing out a
        link between
        twist and change of signature \cite{devastato2018lorentz}, and
         there are indications that a similar change of signature
        occurs for the twist of the spectral triple of the Standard
        Model \cite{Filaci:PhD}. 
\newpage
	
	Here we provide a complementary interpretation,  purely
	geometrical and regardless of any action formula: $f_\mu dx^\mu$ is the
	Hodge dual of a $3$-form (proposition \ref{PropFluctTors}) which, in case the manifold has
        dimension $4$, is the torsion form associated with an
        orthogonal and geodesic preserving torsion (corollary \ref{cor:flucttors}).

A second result of this work is the discovery that all the twisted fluctuations
with exact $1$-form $f_\mu dx^\mu$ are generated by an action of
the group of \emph{twisted unitaries}, that is the elements of the
algebra of the twisted spectral triple which are unitaries with respect to the inner
product induced by the twist (proposition \ref{prop:gentorsion}). In the non-twisted
case it is a major result of \cite{Connes:1996fu} that inner fluctuations
are generated by the unitaries of the algebra. In the twisted case, it
was known that these unitaries could not generate fluctuations. It is
a noticeable result that twisted unitaries do. 

In addition, by extending twisted unitarity to the whole of ${\cal
  B}(L^2(\Man, S))$, one finds the Lorentz group as a proper subgroup
(proposition \ref{prop:Lor}). This strengthens the relation between
twist and change of signature, independently of any action formula.

As a side result, we study the dependance of the fermionic action in the choice of the
unitary that implements the twisting automorphism. We show in
corollary \ref{corr:rgamma0} that the only choice which induces a change of
signature is the one considered in \cite{TwistSpontBreakDevastaMartine2017,
  martinetti2022lorentzian}, namely $R=\gamma^0$ the first Dirac matrix. 
\smallskip
	
	The paper is organised as follows. Section \ref{TSTsec}
        contains generalities on twisted spectral triples, including the real
	structure. We summarise the procedure of minimal twist,
        and apply it to the spectral triple of $\Man$. Section \ref{SecTorsionBases} deals with torsion. It
	contains basic material on contorsion, orthogonal and
	geodesic preserving connections. The first two main results of the paper are in
	section \ref{sec:TwisFlAsSkew}. In \S
        \ref{subsec:twistorsion} we show that in
        dimension $4$,  twisted fluctuations of the Dirac operator
        of $\Man$ 
        yield a skew-symmetric torsion in the spin connection.
This term is invariant under a gauge transformation (\S \ref{subsec:gaugetransf}). 
The
unitaries with respect to the inner product induced by the twist are studied in \S
\ref{subsec:rhoproduct}. 
In \S \ref{subsec:flucgroupact} we show the second main result, namely that a suitable action of twisted
unitaries generates the torsion term. 
Section \ref{sec:actions} deals with the actions, fermionic and
spectral. It contains the results regarding the interpretation of torsion as
energy-momentum, and on the change of signature (\S
\ref{subsec:torsionergy}). The Lorentz group as a subgroup of the
twisted unitaries  is studied in \S\ref{sec:lorentz}, and a spectral
action with torsion is computed in \S\ref{subsec:specact}.
	
\vfill
	
	{\bf Notations:} in all the paper, $\Man$ is a closed
        (i.e. compact without boundary), orientable, riemannian spin
        manifold of even dimension $n=2m$. We use Einstein summation on repeated indices in
	up\slash down positions. Greek indices are for  local
	charts, latin ones high in the alphabet ($a, b...)$
  are for the normal coordinates and the non-local
	orthonormal basis of the tangent bundle $T\Man$  and cotangent
        bundle $T^*\Man$.
	
	In a local chart $\left\{ x^\mu, \mu=1, ..., n\right\}$ on
        $\Man$, we denote $\left\{\partial_\mu, \mu=1, ...,n\right\}$ the
	associated coordinate basis of $T\Man$ and $\left\{ dx^\mu, \mu=1, ..., n\right\}$ the dual basis of
	 $T^*\Man$. We use the abbreviate notations
	$\left\{ x^\mu\right\}, \left\{ \partial_\mu\right\}, \left\{
	dx^\mu\right\}$
	where it is understood that $\mu$ runs on $1, ..., n$. For historical
	reasons, when dealing with the spin structure, the index runs on $0,
	..., n-1$. 

On Minkowski space, $\mu=0$ is the timelike direction and spacelike
directions are labelled by latin indices lower in the alphabet
($i,j...$).	

\newpage
	
	\section{Twisted spectral triples} 
	\label{TSTsec}
	Twisted spectral triples have been introduced in  \cite{TwistTypeIIIConneMosco2006}
with a double motivation: dealing with conformal transformations on
riemannian manifolds, and  applying 
	noncommutative geometry to type III von Neumann algebras. 
	Later on, they found applications~in high energy physics, providing a way to explore models 
	beyond the Standard Model~\cite{TwistSpontBreakDevastaMartine2017}. After recalling the
	main definitions in \S\ref{subsec:realtwsit}, we motivate the interest of twisted spectral triples for gauge theories in
	\S\ref{subsec:twistedfluc}, then introduce  in \S\ref{subsec:miniamltwist}
	our main object of study: the  "minimal" twist of a riemannian closed spin manifold.
	
	\subsection{Real, twisted spectral triples}
	\label{subsec:realtwsit}
	
	\begin{definition}[Connes, Moscovici]
		A twisted spectral triple consists in an involutive algebra
		$\calA$ acting faithfully on a Hilbert space $\calH$, together with a selfadjoint
		operator $\Dir$ on $\calH$ with compact resolvent, and an
		autormorphism $\rho$ of $\calA$ such that for any $a$ in $\cal H$ the \emph{twisted
			commutator}{\footnote{Unless
				needed, we omit
				the symbol $\pi$ of the representation
                                and  identify an
				element of $\A$ with its
                                representation on~$\calH$. The later
                                is always assumed to be involutive: $\pi(a^*)=\pi(a)^\dag$ with
				$^\dag$ the adjoint in ${\cal
                                  B}(\HH)$.
}}
		\begin{equation} [\Dir,a]_\rho\defeq \Dir a-\rho(a)\Dir
			\label{eq:twistcomm}
		\end{equation}
		is bounded. 
		The automorphism is asked to satisfy the regularity condition
		\begin{equation}
			\label{eq:reg}
			\rho(a^*)=\rho^{-1}(a)^*\quad\forall a\in\calA.
		\end{equation}
	\end{definition}
	\noindent 
	One calls  $D$ the (generalised) \emph{Dirac
		operator}. It coincides with the ``true'' Dirac operator~\eqref{eq:freeDirac}
for the canonical spectral triple of a riemannian
	manifold (\eqref{eq:tscano} below).
	
	As in the non-twisted case, a twisted spectral triple is \emph{graded}
	if there exists a selfadjoint operator $\Gamma$ on $\calH$ that squares
	to the identity, $\Gamma^2=\mathbb I$, and such that 
	\begin{equation}
		\label{eq:grad}
		\{D, \Gamma\}=0,\qquad  [\Gamma, a]=0\quad \forall a\in\calA.
	\end{equation}

	The real structure as well is defined as in the non-twisted case \cite{Connes:1995kx}, that
	is an antilinear, unitary
	operator $J$ satisfying
	\begin{equation}
		\label{AxTripl}
		J^2=\epsilon \mathbb \bbbone\qquad\qquad J\Dir =  \epsilon^\prime\Dir J\qquad \qquad J\Gamma =\epsilon^{\prime\prime} \Gamma J
	\end{equation}  
	where $\epsilon, \epsilon', \epsilon''\in\left\{0, 1\right\}$ define the
	$KO$-dimension of the (twisted) spectral triple (see  \S\ref{app:relagrad}).
It satisfies the same \emph{order zero condition} as in the non
twisted case, namely
	\begin{equation}
\label{eq:orderzero}
		[a, Jb^*J^{-1}]=0\quad \forall a, b \in \calA;
	\end{equation}
whereas the \emph{first-order condition} is twisted and becomes
	\begin{equation}
		\label{eq:twistfirstord}
		[[\Dir, a]_\rho,
                Jb^*J^{-1}]_{\rho^\circ}=0\qquad\qquad \forall a, b\in\calA,
	\end{equation}
	where one extends $\rho$ to $J\A J^{-1}$ defining   
	\begin{equation}
		\label{eq:defrhozero}
		\rho(Jb^*J^{-1})\defeq J\rho(b^*)J^{-1}.
	\end{equation}
	 This extension satisfies the same regularity
	condition \eqref{eq:reg} as $\rho$ \cite[eq. 2.6]{Martinetti:2021aa}.

	\subsection{Twisted fluctuation of the metric}
	\label{subsec:twistedfluc}
	
	In describing  gauge theories  like the
	Standard Model in terms of a spectral triples $(\cal A,
	\calH, D)$, the fermionic fields
	are retrieved as element of  the Hilbert space $\calH$. The bosonic
	fields are obtained as 
	\emph{fluctuations of the metric} 
	\cite{Connes:1996fu}.
	This is a process consisting in exporting the spectral triple to an algebra $\cal B$ Morita equivalent to $\calA$. In the
	simplest case of self Morita equivalence, that is $\calB=\cal A$, this amounts to substitute $D$ with the \emph{covariant
		Dirac operator} 
	\begin{equation}
		\label{eq:Dircov}
		D_A:= D + A +\epsilon'JAJ^{-1}
	\end{equation}
	where $A$ is a selfadjoint element of the set of generalised $1$-forms
	\begin{equation}
		\label{eq:bimod}
		\Omega^1(\calA) \defeq\left\{\, \sum_ia_i[\Dir, b_i]\,\, \ \,\, a_i, b_i \in \calA\, \right\}. 
	\end{equation}
	The terminology comes from the abstract construction (well explained in \cite{ConnMarc08b}) in
	which the operator \eqref{eq:Dircov} is the covariant derivative
	associated with a connection on $\calA$ (viewed as a
	module on itself) with value in the bimodule~\eqref{eq:bimod}.

	The  fluctuations of the metric have been adapted to
	the twisted case in
	\cite{TwistSpontBreakDevastaMartine2017,TwistLandiMarti2016,
		TwistGaugeLandiMarti2018}. Given a
	twisted spectral triple $(\cal A, \cal H, D), \rho$ with real
	structure $J$, a twisted fluctuation amounts to
	substitute the Dirac operator with the \emph{twisted-covariant operator}
	\begin{equation}
		\label{eq:twistfluct}
		D_{A_\rho}:= D + A_\rho + \epsilon' J A_\rho J^{-1}
	\end{equation}
	where $A_\rho$ is an element of the set of generalised
	twisted $1$-forms
	\begin{equation*}
		\Omega^1_\Dir(\calA,\rho)\defeq\left\{\, \sum_ia_i[\Dir, b_i]_{\rho}\,\, \ \,\, a_i, b_i \in \calA\, \right\}
	\end{equation*}
	such that \eqref{eq:twistfluct} is selfadjoint.
	
	In the spectral triple of the Standard Model, all the gauge bosons (including the Higgs~\cite{chamseddine2007gravity}) come
	from fluctuations
	\eqref{eq:Dircov}  of the Dirac operator \eqref{eq:Dirac}.
	However there is a
	part $\gamma\otimes D_R$ of this operator that commutes with the algebra. As such, it  is ``transparent
	under fluctuation'' and does not contribute to the generation of bosons. 
	The motivation for twisting the Standard Model was to make $D_R$ fluctuate according to \eqref{eq:twistfluct}, with the hope to obtain the new
	scalar field required to fit the Higgs mass and stabilise the
	electroweak vacuum \cite{chamseddine2012resilience}. 
	
	This was obtained
	in \cite{TwistSpontBreakDevastaMartine2017}
	by twisting the
	electroweak part of the Standard Model. Remaining mathematical problems, stressed in
	\cite{filaci2023critical},
	were later solved in
	\cite{filaci2021minimal} but at the cost of 
	giving up the first-order condition \eqref{eq:twistfirstord}, in a similar way as what
	is done for the
	non-twisted case in \cite{chamseddine2013inner}. The last paper
	actually shows how abandoning the
	first-order condition (for  a
	usual spectral triple) is enough
	to get the required extra scalar field. So it seemed there were no more 
	added-value in twisting.
	
	Nevertheless, independently of the first-order condition, the twist of
	the Standard Model  also yields 
	an unexpected new field of
	$1$-forms, coming from the twisted fluctuation of the free part
	$\slashed\partial\otimes{\mathbb I}_F$ of the
	operator \eqref{eq:Dirac}.
	Besides one examples stressed in the introduction, the general meaning of this field was not clear so far. In this paper we provide a geometrical
	interpretation, in term of torsion in the spin connection. The analysis
	does not depend on the details of the finite dimensional
	part of the spectral triple, but only on the manifold part. This is why in
	the following we restrict to the twisted spectral triple of a manifold.
	
	\subsection{Minimal twist of a manifold}
	\label{subsec:miniamltwist}
	
	So far, in all the applications to high energy physics,  twisted
	spectral triples are obtained by minimally twisting an existing
	spectral triple $({\mathcal A}, {\mathcal H}, D)$. By ``minimal
	twist'' one intends that the Hilbert space and the
	Dirac operator are untouched, only the
	algebra is modified. Physically this means that 
 the fermionic
	content of the model (encoded within $\cal H$ and $D$) is conserved. One only looks at new bosons.
	
	Such minimal twists are easily obtained if the spectral
	triple is graded. Indeed, the properties \eqref{eq:grad} of the grading $\Gamma$
	guarantee that $\mathcal H$ carries two independent
	representations of $\mathcal A$, one on each eigenspace of $\Gamma$. Moreover,  the twisted commutator~\eqref{eq:twistcomm} is bounded for  $\rho$ the automorphism
        that flips the two copies of
	${\mathcal A}$ \cite[Prop. 3.7]{TwistLandiMarti2016}.
	
	Explicitly, starting with the canonical spectral triple of an
        oriented, 
        riemannian, closed, spin manifold $\Man$ of even dimension $n=2m$, namely
	\begin{equation}
		\label{eq:tscano}
		(C^\infty(\Man), L^2(\Man, \Sp), \slashed\partial)
	\end{equation}
	where
	the algebra $C^\infty(\Man)$ of smooth functions on $\Man$ act by
	multiplication on the Hilbert space $L^2(\Man, S)$ of square
	integrable spinors on $\Man$ and $\slashed\partial$ is the Dirac
	operator \eqref{eq:freeDirac}, together with grading $\gamma$
        and real structure $\cal J$ (whose explicit form is given in
        appendix), 
        then one obtains the twisted spectral triple
	\begin{equation}
		\label{eq:canon}
(C^\infty(\Man)\otimes \mathbb{C}^2, L^2(\Man, \Sp),
                \slashed\partial), \rho 
	\end{equation}
	with twist
        \begin{equation}
\label{eq:minimantwist}
\rho(f,g)= (g,f) \quad \forall (f,g)\in C^\infty(\Man)\otimes
	\mathbb C^2.
      \end{equation}
It has the same grading and real structure as \eqref{eq:tscano} \cite[Prop.  3.8]{TwistLandiMarti2016}.
	
	It is instructive to check the boundedness of the twisted commutator.
	The two copies of $C^\infty(\Man)$ act independently  on left\slash
	right  components of
	spinors, that~is
	\begin{equation}
		\label{FormOfa}
		\pi(a)=\begin{pmatrix}
			f \mathbb \bbbone_{2^{m-1}} & 0  \\
			0 &  g\mathbb \bbbone_{2^{m-1}}  	\end{pmatrix}
	\end{equation}
	where $2^{m}$ is the dimension of the spin representation and we
	split $L^2(\Man, S)$ into the direct sum of the two eigenspaces of $\gamma$ (i.e. the chiral base). Using remark \ref{eq:spincommut}
	below together with \eqref{eq:twistcomgpi}, one checks that the
	twisted commutator with $\slashed\partial$ is bounded: 
	\begin{align}
		\left[\slashed\partial, \pi(a)\right]_\rho&=-i\left(\gamma^\mu
		{\widetilde\nabla}_\mu^S \pi(a)- \pi(\rho(a))\gamma^\mu{\widetilde\nabla}_\mu^S\right)=
		-i\gamma^\mu \left[{\widetilde\nabla}_\mu^s, \pi(a)\right] ,\\
		&= -i\gamma^\mu
		\left[\partial_\mu,
		\pi(a)\right]
		=-i\gamma^\mu \begin{pmatrix}
			(\partial_\mu	f) \mathbb \bbbone_{2^{m-1}} & 0  \\
\label{eq:rhocommutator}			0 &  (\partial_\mu g)\mathbb \bbbone_{2^{m-1}} \end{pmatrix}	\,.
	\end{align}
	
	A twisted fluctuation generates a non-zero selfadjoint
	term  $A_\rho+\epsilon^\prime \calJ A_\rho\calJ^{-1}$ only in $KO$-dimension
	$0$ and $4$ \cite[Prop 5.3]{TwistLandiMarti2016}. Then
the twisted-covariant operator
        \eqref{eq:twistfluct}~is 
	\begin{equation}
		\label{TwistFluctDir}
		{\ds}_{A_\rho}= \ds-i\gamma^\mu f_\mu \gamma
	\end{equation}
	where the $f_\mu$'s are smooth real functions on $\Man$
        (details are recalled in \S\ref{app:relagrad}). 

	The aim of this paper is to study the additional term $-i\gamma^\mu f_\mu \gamma $.
We give a geometric interpretation in proposition \ref{PropFluctTors} below,
in particular as a torsion (corollary \ref{cor:flucttors}). We also
        show in proposition \ref{prop:gentorsion}  how to generate this additional term through a suitably twisted action
	of a group of unitaries.
	
	\section{Torsion}
	\label{SecTorsionBases}
	
	We recall some properties of the torsion of a connection
	$\nabla$ on a riemannian manifold,  in particular when
        $\nabla$ is orthogonal
	(i.e. compatible with the metric) in \S \ref{subsec:ortho}, or
        has the same geodesics as the Levi-Civita connection
        (\S\ref{subsec:geopres}). Both conditions yields the definition
        of the torsion $3$-form in \S \ref{subsec:threeform}. The lift to spinors is studied in \S\ref{SubsecSpinMfld}.
	All in this section are classical results, but the proofs are
        no always so easy to find in the literature, that is why we
        prefer to give them explicitly. Good references
        are \cite{Lee:2010aa} and \cite{Lawson:1989aa}

\subsection{Orthogonal connection and contorsion}
	\label{subsec:ortho}

Recall that a connection on the tangent bundle $T\Man$ 
	of a differential manifold $\Man$ is a map 
	\begin{eqnarray}
		\nabla:& T\Man \times T\Man &\longrightarrow\; T\Man,\\
		&X,\; Y&\longmapsto\; \nabla_XY
	\end{eqnarray}
	$C^\infty(\Man)$-linear in the first entry and satisfying the
	Leibniz rule in the second.
	Its torsion  is the $(2, 1)$-tensor field ($[\cdot, \cdot ]$ denotes the Lie bracket)
	\begin{eqnarray}
		T:& T\Man \times T\Man &\longrightarrow T\Man,\\
		& X,\, Y&\longmapsto \nabla_XY-\nabla_YX-[X,Y].
	\end{eqnarray}
	
	Given a metric $g$ on  $\Man$, a connection $\nabla$ is \emph{metric} (or orthogonal) if 
	\begin{equation}
		\label{eq:mtriccomp}
		X[ \langle Y , Z
		\rangle]=\langle \nabla_XY , Z \rangle + \langle Y , \nabla_XZ
		\rangle\quad \forall X,Y,Z\, \in \,T\Man
	\end{equation}
	with 
	$\langle X , Y
	\rangle :=g(X, Y)$
	the inner product on
	$T\Man$ defined by the metric.
	By Levi-Civita theorem,  there exists  a unique orthogonal
	connection $\widetilde{\nabla}$ with vanishing torsion. 
	The difference between any two connections is a $(2,1)$ tensor
	field. It is customary to call \emph{contorsion} the difference with the Levi-Civita connection.
	\begin{definition}
		\label{def:cont} The contorsion of a connection $\nabla$ is the $(2, 1)$-tensor
		field
                \begin{equation}
K= \nabla-~\widetilde{\nabla}.
\end{equation}

	\end{definition}
	
	The orthogonality of a connection can be read in the properties of the
	$(3, 0)$ tensor
	\begin{align}
		\label{eq:mapKbis}
		\Kont^\flat (Z,X,Y): & \quad T\Man\times T\Man\times T\Man\to
		C^\infty (\Man),\\
		&\quad\quad X,\quad Y,\quad Z \quad \mapsto \langle Z, \Kont(X,Y) \rangle.
	\end{align}
	
	\begin{proposition}
		\label{prop:metric}
		A connection $\nabla$ is orthogonal iff $\Kont^\flat$
		is skew-symmetric in $Z$ and~$Y$.
	\end{proposition}
	\begin{proof}
		From the definition of $K^\flat$ and the symmetricity of the metric, one has
		\begin{align}
			\label{eq:mapK}
			K^\flat(Z, X, Y)&= \langle \nabla_XY,Z\rangle  -\langle
			\widetilde\nabla_XY,Z\rangle.
		\end{align}
		If $\nabla$ is orthogonal, subtracting the metric condition
		\eqref{eq:mtriccomp} for $\widetilde\nabla$ from the one of $\nabla$
		gives 
		\begin{align}
			0= 	\left(\langle {\nabla}_XY , Z \rangle - \langle
			\widetilde{\nabla}_XY , Z \rangle\right) +\left(\langle {\nabla}_XZ , Y
			\rangle - \langle \widetilde{\nabla}_X Z , Y \rangle\right), 
		\end{align}
		that is
		\begin{equation}
			\Kont^\flat(Z,X,Y)+\Kont^\flat(Y,X,Z)=0.
		\end{equation}
		
		Conversely, assuming the last equation, then \eqref{eq:mapK} yields
		\begin{align*}
			\langle \nabla_XY , Z \rangle + \langle Y , \nabla_XZ \rangle= \langle \widetilde{\nabla}_XY , Z \rangle + \langle Y , \widetilde{\nabla}_XZ \rangle 	=	X [\langle Y , Z \rangle]
		\end{align*} 
		where the last equality follows from the orthogonality of the Levi-Cevita
		connection. Hence $\nabla$ satisfies \eqref{eq:mtriccomp} .
	\end{proof}

	\subsection{Preservation of the geodesics}
\label{subsec:geopres}
	
	Connections with the
	same geodesics as the Levi-Civita one are particularly
	relevant, for they may yield modification of general relativity that do
	not alter the
	results based on geodesics.
	
	\begin{proposition}
		\label{GeodSkew}
		A connection $\nabla$ has the same geodesics as the Levi-Civita
		connection $\widetilde\nabla$ if and only if its contorsion $K$ is antisymmetric 
		\begin{equation}
			\Kont(X,Y)=-\Kont(Y,X).
		\end{equation}
	\end{proposition}
	\begin{proof}
		Assume $K$ antisymmetric. Then $K(X, X)=0$
		for any $X\in T\Man$. For $X$ tangent to a geodesic
		of $\nabla$, that is $\nabla_X X=0$, then Def. \ref{def:cont}
		yields $\widetilde \nabla_X X=0$, meaning that $X$ is tangent to a
		geodesic of $\widetilde\nabla$. Similarly, any vector field tangent to a
		geodesic of $\widetilde\nabla$ is tangent to a geodesic of $\nabla$.
		In other terms the two connections have the same geodesics.
		
		Conversely, assume $\nabla$ has the same geodesics as
		$\widetilde\nabla$ and fix $p\in\Man$. In the normal coordinates in $p$, the geodesics through
		$p$ (for both $\nabla$ and $\widetilde\nabla$)  are all the straight lines
		$  t\mapsto (V^1t, ..., V^n t)$  with $V^{a=1, ...,n}$ arbitrary real constants,
		not all
		simultaneously vanishing. The corresponding 
		geodesic equations for $\nabla$ and $\widetilde\nabla$, written in $p$, are
		\begin{equation}
			\Gamma^c_{ab}(p) V^a V^b =0, \quad
			\widetilde\Gamma^c_{ab}(p) V^aV^b =0\quad\forall
			c=1, ...,n
		\end{equation}
		where $\Gamma^c_{ab}$, $\tilde\Gamma_{ab}^c$ are the components of
		$\nabla$, $\tilde\nabla$ in the
		normal coordinates.
		
		For a geodesic tangent in $p$ to a vector  with only one
		non-zero component, say  $V^a=1$, one gets
		\begin{equation}
			\label{eq:geomumu}
			\Gamma^c_{aa}(p)=0=  \widetilde \Gamma^c_{aa}(p)
			\quad\forall c =1, ..., n.
		\end{equation}
		Then, for a geodesic tangent to a vector  with only two non-zero components
		$V^a=V^b=1$, the geodesic equations yield
		\begin{equation}
			\Gamma^c _{ab}(p)  +  \Gamma^c _{ba}(p) 
			=0 =\widetilde\Gamma^c _{ab}(p) +  \widetilde\Gamma^c _{ba}(p)  \quad\forall c=1, ..., n,
		\end{equation}
		that is
		\begin{equation}
			\label{eq:vdeux}
			\Gamma^c _{ab}(p) -\widetilde\Gamma^c _{ab}(p)
			= -  \left(\Gamma^c _{ba}(p) -\widetilde\Gamma^c
			_{ba}(p)\right) \quad\forall c=1, ..., n.
		\end{equation}
		
		In a local chart $\{x_\mu\}$ with associated
		basis $\{\partial_\mu\}$ of $T\Man$ ($\mu=1, \dots n$), the tensor $K$ has
		components
		\begin{equation}
			\label{eq:composante}
			\Kont^\lambda_{\mu\nu}\defeq\langle
			\Kont(\partial_\mu, \partial_\nu) , dx^\lambda \rangle=
			\Gamma_{\mu\nu}^\lambda - \widetilde\Gamma_{\mu\nu}^\lambda .
		\end{equation}
		Similarly, the right hand side of \eqref{eq:vdeux} are the components
		of $K$ in the normal coordinates. Together with \eqref{eq:geomumu},
		this  shows that $K_{ab}^c(p)
		= - K_{ba}^c(p)$ for any $c$. Since $a, b$ and $p$ are arbitrary,  $K$ is antisymmetric.
	\end{proof}

	Geodesic preservation also reads in the relation between
	torsion and contorsion.
	\begin{corollary}
		A connection $\nabla$ as the same geodesic as the Levi-Civita one
		if, and only if,  it has  torsion $T=2\Kont$.
	\end{corollary}
	\begin{proof}
		\label{lem:contorsion}  
		The Levi-Civita connection being torsionless, one has
		$\widetilde\nabla_X Y = \widetilde\nabla_YX + [X, Y]$. Thus
		\begin{align}
			\Kont(X,Y)-\Kont(Y,X) &= (\nabla_XY - \widetilde\nabla_XY)-  (\nabla_YX-
			\widetilde\nabla_YX),\\
			&= \nabla_XY - \nabla_Y X  - [X,Y]= T(X,Y).\nonumber
		\end{align}
		The result then follows from proposition \ref{GeodSkew} \end{proof}
	
	\subsection{Torsion $3$-form}
\label{subsec:threeform}
	
	As stressed above,  it is reasonable to assume
	that any physically acceptable connection has the same
	geodesics as the Levi-Civita connection. One may also desire to keep
	the compatibility with the metric \eqref {eq:mtriccomp}, so that the
	pseudonorm of (timelike) vector is invariant under parallel transport. Therefore, good
	candidate to alternative theory of gravity are the connections which
	are both orthogonal  and geodesics preserving. 
	
	\begin{proposition}
		\label{Totalskew}
		A connection $\nabla$ is orthogonal and geodesic preserving iff it contorsion is such that $K^\flat$ is totally antisymmetric.
	\end{proposition}
	\begin{proof}
		Assume $\nabla$ is orthogonal and geodesic preserving. Then
                \begin{equation}
\label{eq:metrorto}
                  \Kont^\flat(Z,X,Y)=-K^\flat(Y, X, Z)\quad \text{ and }
                  \quad \Kont^\flat(Z,
		X, Y)=-K^\flat(Z, Y, X)
                \end{equation}
		by the metric property of Prop.~\ref{prop:metric} and
                geodesic preservation of
		Prop.~\ref{GeodSkew}. The skew symmetry in $Z,X$ follows from 
		\begin{align*}
			\Kont^\flat(Z, X,Y)=-\Kont^\flat(Z, Y, X)=\Kont^\flat(X , Y, Z)=-\Kont^\flat(X ,Z ,Y).
		\end{align*} 
		Conversely, if $\Kont^\flat$ is totally skew symmetric then in particular it
		satisfies \eqref{eq:metrorto}.
	\end{proof}
	
	Let $\Omega_{MG}$ denote the set of orthogonal and geodesic preserving
	connections. The proposition above shows that the tensor $K^\flat$ associated to any $\nabla\in
	\Omega_{MG}$ is a
	$3$-form, called the \emph{torsion $3$-form}.

	\subsection{Lift to spinors} 
	\label{SubsecSpinMfld}
	
	An orthogonal connection on the tangent bundle of a
	riemannian manifold $(\Man, g)$ can be lifted to
	the spinor bundle  as soon as the second Whitney class of $\Man$ 
	vanishes. The construction passes through the principal bundle of
	frames, where the lift from the orthogonal group $SO(n)$ to its
	double cover $\text{Spin}(n)$ actually occurs. This explains why the
	lift to spinors is
	considered for orthogonal connections only, and uses the 
	orthonormal sections (see appendix \ref{sec:aporto})  
	\begin{equation}
		\label{eq:baseortdeux}
		\left\{ E_a, a=1, ..., n\right\} 
	\end{equation}
	of the frame bundle, in which the metric $g$ is diagonal. 
	
	More precisely, the lift to the spinor bundle of an orthogonal connection $\nabla$ on
	$\Man$~is:
	\begin{equation}
		\label{eq:CovSpin}
		\nabla_\mu^S=\partial_\mu +\frac 14 \Gamma_{\mu a}^b
		\gamma^a \gamma_b
	\end{equation}
	where $\gamma^{a=1, ..., n}$  are the euclidean Dirac matrices (see appendix
	\ref{ConvGamMat}) , $\gamma_b=\delta_{ab}\gamma^a$ with $g_{ab}$ the
	components of the metric in the orthonormal frame \eqref{eq:baseort}
	and the $\Gamma_{\mu a}^b$'s are the components of $\nabla_\mu E_a$ in the orthonormal frame, namely 
	\begin{equation}
		\label{eq:connmix}
		\nabla_\mu E_a = \Gamma_{\mu a}^b E_b,
	\end{equation}
	
	\newpage
	 To see the relation with the contorsion tensor, it is useful to work
         out the expression of the spin connection
	\eqref{eq:CovSpin} in a local chart. 
	
	\begin{proposition}
		\label{Prop:gammaloal}
		
		One has
		\begin{equation}
			\Gamma_{\mu a}^b \gamma^
			a\gamma_b=\left( \Gamma_{\mu\nu}^\rho g_{\rho\lambda}
			- g_{ab} e^b_\lambda \partial_\mu e^a_\nu \right) \gamma^\nu\gamma^\lambda
		\end{equation}
		where the (inverse) \emph{vielbein}
		$e^a_\mu\in C^\infty(\Man)$ for $a, \mu =1, ..., n$ are the
		coefficients of the coordinate basis in the non local frame:
		$\partial_\mu = e_\mu^a E_a $.
	\end{proposition}
	\begin{proof}
		Computing $\nabla_\mu\partial_\nu$ in the non-local basis in the
		two following way, 
		\begin{align}
			\nabla_\mu\partial_\nu &=\Gamma_{\mu\nu}^\rho \partial_\rho=
			\Gamma_{\mu\nu}^\rho e_\rho^b E_b,\\
			\nabla_\mu\partial_\nu &=\nabla_\mu(e_\nu^b E_b)= (\partial_\mu
			e_\nu^b) E_b + e_\nu^b \Gamma_{\mu b}^a E_a =\left( \partial_\mu
			e_\nu^b+ e_\nu^a \Gamma_{\mu a}^b\right) E_b 
		\end{align}
		one obtains
		\begin{equation}
			e_\nu^a \Gamma_{\mu a}^b =  \Gamma_{\mu\nu}^\rho e_\rho^b - \partial_\mu
			e_\nu^b.
		\end{equation}
		The Dirac matrices in a local chart are defined as $\gamma^\mu =e^\mu_a
		\gamma^a$, that can be inverted as 
		$\gamma^a=e^a_\nu\gamma^\nu$.
		As well, 
		$\gamma_b = g_{ba}\gamma^a = g_{ba} e^a_\lambda \gamma^\lambda$.
		Thus\begin{equation}
			\Gamma_{\mu a}^b \gamma^
			a\gamma_b= \Gamma^b_{\mu a} e^a_\nu \gamma^\nu
			g_{ba}e^a_\lambda\gamma^\lambda=
			\left(\Gamma_{\mu\nu}^\rho e_\rho^b - \partial_\mu
			e_\nu^b\right) g_{ba}e^a_\lambda\gamma^\nu \gamma^\lambda.
		\end{equation}
		The result follows from
		\begin{equation}
			g_{ba}e^b_\rho e^a_\lambda = g(e^b_\rho E_b, e^a_\lambda E_a)=
			g(\partial_\rho, \partial_\lambda) =g_{\rho\lambda},
		\end{equation}
		and exchanging the indices $a$ and $b$ in the second term.
	\end{proof}
	
	From the definition \ref{def:cont} of the contorsion and
	Prop. \ref{Totalskew}, in a local chart any connection
	$\nabla\in\Omega_{MG}$ has components
	\begin{equation}
		\label{eq:compskew}
		\Gamma^\rho_{\mu\nu} = \widetilde   \Gamma^\rho_{\mu\nu}  +K^\rho_{\mu\nu},
	\end{equation}
	such that the components
	\begin{equation}
		K_{\lambda\mu\nu}=g_{\lambda\rho}K^\rho_{\mu\nu}
	\end{equation}
	of the torsion $3$-form are totally skew-symmetric.
	By \eqref{eq:connmix} and Prop. \ref{Prop:gammaloal}, the lift of $\nabla$ to spinors is thus
	\begin{equation}
		\label{eq:lifttorsion}
		\nabla_\mu^S = \widetilde\nabla_\mu^S + \frac 14
		K^\rho_{\mu\nu}g_{\rho\lambda}\gamma^\nu\gamma^\lambda
		=\widetilde\nabla_\mu^S  + \frac 14 K_{\nu\lambda\mu}\gamma^\nu\gamma^\lambda
	\end{equation}
	where
	\begin{equation}
		\label{eq:liftLV}
		\widetilde\nabla_\mu^S = \partial_\mu +\widetilde
		\Gamma^b_{\mu a}\gamma^a\gamma_b
	\end{equation}
	is the lift of the Levi-Civita connection discussed in
	the introduction, and for the second equality we use
	$K_{\lambda\mu\nu}=K_{\mu\nu\lambda}$ following from the antisymmetry
	of $K^\flat$.  
	
	\begin{remark}
		\label{eq:spincommut}
		One checks from \eqref{eq:liftLV}  that $[\widetilde\nabla_\mu^S,
		f]=[\partial_\mu, f]$ since $f$ acts by multiplication on
		spinors, so commutes with all the Dirac matrices.
	\end{remark}
	
	\newpage
	\section{Torsion for minimally twisted manifolds}
	\label{sec:TwisFlAsSkew}
	
	In this section we show the first main results of this paper, namely that the
	minimal twist of an oriented,  closed, riemannian spin
        manifold $\Man$ of dimension $2m=4$
	induces a orthogonal and geodesic preserving torsion
        (corollary \ref{cor:flucttors}). 

We obtain first
	a more general result, valid for any even dimension and
        for $KO$-dimensions $0$ and $4$,
        which explains the link between the $1$-form $f_\mu dx^\mu$
        and the term $f_\mu\gamma^\mu \gamma$ in the twisted
        fluctuation: because  of the presence of the grading $\gamma$,
        this term is not the Clifford action of the $1$-form, but
        of its Hodge dual (proposition \ref{PropFluctTors}). 

In \S\ref{subsec:gaugetransf} we show that the torsion term is
        gauge invariant.

	\subsection{Twisted fluctuation as torsion}
	\label{subsec:twistorsion}
	
	Let us begin by a technical lemma showing that the product of the
	grading $\gamma$ \eqref{eq:gammpermut} by any Euclidean Dirac
        matrix \eqref{EDirac} results in the absorption of the later.
	\begin{lemma}
		\label{lemma:gammagamma}
		Let $\Man$ be of dimension $2m$. For any  fixed value of $a$ in $[0, 2m-1[$,  one has
		\begin{equation}
			\label{EqTechni}
			\gamma^a\gamma= -\frac{(-i)^m}{(2m)!}\epsilon_{a\,a_1\dots
				a_{2m-1}} \, \gamma^{a_1}\dots\gamma^{a_{2m-1}}.
		\end{equation}
	\end{lemma}
	\begin{proof}
		Fix a value  $a$ in $[0, 2m-1]$. For any non-zero term of
		the sum 
		$$ \epsilon_{b_1\dots b_{2m}}\gamma^{b_1}\dots \gamma^{b_{2m}}\gamma^a$$
		the indices $b_{i=0, ..., 2m-1}$ are all distinct, so
		there is one and only one of them - say $b_{i_a}$ -  such that $b_{i_a}=a$. Therefore
		\begin{align}
			\gamma^a(\epsilon_{b_1\dots b_{2m}}\gamma^{b_1}\dots \gamma^{b_{2m}})&=
			(-1)^{i_a}\,\epsilon_{b_1\dots b_{2m}}\,\gamma^{b_1}\dots
			\gamma^{b_{i_a}-1}
			(\gamma^a)^2
			\,
			\gamma^{b_{i_a}+1}
			\dots\gamma^{b_{2m}},\\
			&=\epsilon_{a\,b_1\dots b_{{a_l}-1} b_{{i_a}+1}\dots b_{2m}  }\,\gamma^{b_1}\dots
			\gamma^{b_{i_a}-1}
			\gamma^{b_{i_a}+1}
			\dots\gamma^{b_{2m}},\\
			&=\epsilon_{a\,a_1\dots \dots a_{2m-1}}\,\gamma^{a_1}\dots
			\gamma^{a_{2m-1}},
		\end{align}
		where we use that any euclidean Dirac matrix square to the identity,
		anticommutes with the other ones 
		and
		\begin{equation}
			(-1)^{i_l}\epsilon_{b_1\dots b_{2m}}=
			\epsilon_{b_{i_a}\,b_1\dots b_{{i_l}-1} b_{il+1}\dots b_{2m}},
		\end{equation}
		then redefine the indices as
		\begin{equation}
			a_i:=\left\{
			\begin{array}{ll}
				b_i&\text{ for } i=1, ..., i_l-1,\\ 
				b_{i+1}&\text{ for } i=i_l+1, ..., 2m . \end{array}\right.  
		\end{equation}
		The result follows multiplying on the left the expression \eqref{eq:gammpermut} of the
		grading by $\gamma^a$.
	\end{proof}

           To have the index $a$ in the same position on both
          sides of \eqref{EqTechni}, one writes
          $\epsilon_{a a_1 ... a_n}$ as
          \begin{equation}
            \epsilon^a_{a_1 ... a_{2m}}:=\delta^{ab}\epsilon_{b a_1 ... a_n} .  
          \end{equation}
     
\newpage
	The next proposition gives the geometrical interpretation of 
        the additional term in the  twisted covariant Dirac operator 
	 \eqref{TwistFluctDir}.
	\begin{proposition}
		\label{PropFluctTors}
		In  $KO$-dimensions $0$ and $4$, one has
		\begin{equation}
\label{eq:twisterm}
			i\gamma^\mu f_\mu\gamma =\frac{(-i)^{m+1}}{(2m)}  \, c(\star\omega_f)
		\end{equation}
	\end{proposition}
	where $c$ is the Clifford action~\eqref{eq:Cliffaction} and $\star\omega_f$ is the Hodge dual of the $1$-form
        \begin{equation}
\label{eq:defomega}
        \omega_f=f_\mu dx^\mu.
      \end{equation}
\begin{proof}
		We work in orthonormal coordinates, absorbing the vielbein in the
		component $f_\mu$ of the twisted fluctuation by defining
		$ f_a:= e^\mu_a f_\mu$, so that 
			$\gamma^\mu f_\mu=e^\mu_a f_\mu \gamma^a = f_a\gamma^a$.
		By lemma \ref{lemma:gammagamma} one has 	\begin{align}
			\gamma^\mu f_\mu \gamma=f_a\gamma^a \gamma
			&=\frac{(-i)^m}{2m}  \frac 1{(2m-1)!} f_a\,\,\delta^{ab}\epsilon_{b\,
				b_1\dots  b_{2m-1}}\gamma^{b_1}\dots
			\gamma^{b_{2m-1}},\\
			\label{eq:starfgamma}
			&= \frac{(-i)^m}{2m} (\star \omega_f)_{b_1  \dots
				b_{2m-1}}\gamma^{b_1}\dots \gamma^{b_{2m-1}}\\
			&=\frac{(-i)^m}{2m}  c(\star\omega_f)
		\end{align}
		where we use \eqref{eq:Hodge} for the components of the Hodge dual.
	\end{proof}
	
From now on, we denote the twisted covariant Dirac
operator~\eqref{TwistFluctDir} as
\begin{equation}
\label{eq:domegaf}
  \ds_{\omega_f}=\ds - if_\mu\gamma^\mu\gamma.
\end{equation}

	In
	dimension $4$, proposition \ref{PropFluctTors}  has an interpretation in term of torsion, 
	\begin{corollary}
		\label{cor:flucttors}
       For $\cal M$ of dimension $4$, the twisted covariant Dirac operator $\ds_{\omega_f}$ is the lift to spinors of an
		orthogonal and geodesic preserving connection, with torsion $3$-form $-\star\omega_f$.
	\end{corollary}
	\begin{proof} For $m=2$, \eqref{eq:starfgamma} yields
		\begin{equation}
			-i\gamma^\mu f_\mu \gamma= i\frac 14 (\star\omega_f)_{\mu\nu\rho}\gamma^\mu\gamma^\nu\gamma^\rho
			=-i\gamma^\mu\left(-\frac 14(\star\omega_f)_{\nu\rho\mu}\gamma^\nu\gamma^\rho\right).
		\end{equation}
		Therefore $\ds_{\omega_f}$ is the  Dirac
		operator associated with the connection
		\begin{equation}
			\nabla^\mu = \tilde\nabla_\mu^S + \left(-\frac
                          14(\star \omega_f)_{\nu\rho\mu}\gamma^\nu\gamma^\rho\right).
		\end{equation}
		From   \eqref{eq:lifttorsion}, this is the lift to spinors of a
		connection whose torsion $3$-form has components $-(\star\omega_f )_{\nu\rho\mu}$,
		that is $K^\flat = -\star\!\omega_f$.
	\end{proof}
	\begin{remark}
	The additional term \eqref{eq:twisterm}	does not altered twisted $1$-form: 
		from \eqref{eq:twistcomgpi} one has
		\begin{align*}
			[\ds_{\omega_f},a]_\rho=[\ds,a]_\rho-if_\mu
			(\gamma^\mu \gamma a - a\gamma^\mu\gamma)=
                  [\ds, a]_\rho-if_\mu
			(\rho(a) - a)\gamma^\mu \gamma)=[\ds, a]_\rho.
		\end{align*}
		Thus,  if one equips the space of pure states of
		$C^\infty(\Man)\otimes \mathbb C^2$ (made of two copies of $\Man$)
		with the spectral distance \cite{Connes:1992bc} (see
                also \cite{Martinetti:2016aa}) in which the commutator is substituted with
		a twisted-commutator, then the distance will be invariant under the
		adjunction of the additional term. In dimension $4$,
                this is coherent with the fact that the corresponding torsion is 
		geodesic preserving, hence should not alter the riemannian distance
		between points. 
	\end{remark}

	\subsection{Gauge transformation}
	\label{subsec:gaugetransf}

	A gauge transformation, in the framework described in \S
	\ref{subsec:twistedfluc}, is a change of connection in the module $\cal E$
	that implements Morita equivalence between  $\A$ and
	$\cal B$, induced by a unitary endomorphism of $\cal E$. 
	In case of self Morita equivalence - which is the one
	we are interested here - unitary endomorphisms are in $1$-to-$1$
	correspondance with the unitary elements of $\A$, which form the group\footnote{We restore the symbol of representation to stress that
		the identity holds in ${\cal B}({\cal H})$, and not necessarily in
		$\A$ if the algebra is not uital.}
	\begin{equation}
\label{eq;unitaries}
		{\cal U}( \A):=\left\{u\in \A, \; \pi(u^*)\pi(u)=\pi(u)\pi(u^*)=\mathbb \bbbone\right\}. 
	\end{equation}
	A change of connection induces
	the substitution in the covariant Dirac operator, 
	\eqref{eq:twistfluct} of $A_\rho$ with \cite[Prop- 4.3]{TwistGaugeLandiMarti2018}
	\begin{equation}
\label{eq:twistgaugepot}
		A_\rho^u := \rho(u) [D, u^*]_\rho + \rho(u) A_\rho u^* .
	\end{equation}
	This is a twisted version of the noncommutative version of the usual
	formula of transformation of a gauge potential. We thus call
        $A_\rho$ the \emph{twisted gauge potential}.
	
	Such gauge transformations are obtained by a
        suitably twisted action of~${\cal U}(\A)$.
 First, one defines the \emph{adjoint action}
          of unitaries as
          \begin{equation}
            \label{eq:adgrep}   \text{Ad}(u) \psi :=
            u\, J u J^{-1}\psi \qquad \forall \psi\in \HH, \;
            u\in {\cal U}(\A).
          \end{equation}
One then shows that a twisted gauge
transformation \eqref{eq:twistgaugepot} is equivalent to the (twisted)
conjugate action of $\text{Ad}(\cal U)$, namely  \cite[$\S A$]{devastato2018lorentz} and \cite[Prop. 4.5]{TwistGaugeLandiMarti2018}) 
	\begin{equation}
		\label{eq:conjugactt}
		D_{A_\rho^u} =  \text{Ad}(\rho(u)) \, D_{A_\rho} \,{\text{Ad}(u)}^{-1}.
	\end{equation}

      	All these formulas are the twisted version of their non-twisted
	counterparts, introduced in \cite{Connes:1996fu} (see also
	\cite{chamseddine2007gravity} and \cite{ConnMarc08b} for more
	details). In the spectral description of the Standard Model, they
	give back the gauge transformation of the bosons. 
	The same is true 
	for the twisted spectral triple of the Standard Model developed in
	\cite{filaci2021minimal}, as well as for the twisted spectral triple of
	electrodynamics \cite{martinetti2022lorentzian}. 

However, in both examples the
	additional $1$-form field $\gamma^\mu f_\mu \gamma$ is invariant
	under gauge transformations. This had already been established in full
	generality in
	\cite{TwistLandiMarti2016}, but it takes a new signification now that this field identifies with a torsion (at least in dimension $4$), so we
	restate it as the following proposition.

	\begin{proposition}
		\label{prop:gaugeinv}
The operator $\ds_{\omega_f}$ is gauge
                invariant.
	\end{proposition}
	\begin{proof}
		For the minimal twist of a manifold one has $\hat u=u^\dagger$ for any  $u\in
		{\cal U}(\cal A)$ (cf \linebreak \cite[Lemma 5.1]{TwistLandiMarti2016}). Hence
                \begin{equation}
\label{eq:adunit}
		\text{Ad}(u)=\hat u u =u^\dagger u =\I.
              \end{equation}
		For
		an autormorphism $\rho$ such that $\rho^2=\mathbb \bbbone$ (as the flip), the regularity
		condition \eqref{eq:reg} guarantees that $\rho$ is a
		$*$-automorphism. Thus $\rho(u)$ is also unitary, and
		$\text{Ad}(\rho(u))$ is the identity.  Hence the right-hand-side of
		\eqref{eq:conjugactt} for $D_{A_\rho}=\ds_{\omega_f}$ is
		$\ds_{\omega_f
}$ itself. 
	\end{proof}

	The invariance of $\ds_{\omega_f}$ under a gauge transformation
	\eqref{eq:conjugactt} applies in particular to~$\ds$. This means that
	$\ds_{\omega_f}$ cannot be generated by  gauge transformations of
	$\ds$~itself, in contrast with the  fluctuations (twisted or
	not) of the Dirac operator of the Standard Model: its gauge
	transformations generate
	some  (even if not all) fluctuations.

	\subsection{Twisted unitaries}
	\label{subsec:rhoproduct}

As stressed above, the torsion term does not arise as a gauge
transformation of the Dirac operator. Said differently, the twisted conjugate
action \eqref{eq:conjugactt} of the unitary group does not generate
torsion. However there is a class of torsion - those with co-exact
$3$-form -  which is generated by the action of the group of \emph{twisted
  unitaries}. This is shown in \S \ref{subsec:flucgroupact}. In this
section we recall the definition of  twisted unitaries.

	The twisting
	automorphism  \eqref{eq:minimantwist} of the minimally twisted
        even dimensional manifold coincides
        with the inner automorphism of ${\cal B}(L^2(\Man, S))$ -
        still denoted $\rho$ - induced by the
        first Dirac matrix $\gamma^0$ (in the unitary
        representation \eqref{EDirac}), 
 namely
\begin{equation}
  \label{eq:rhoo}\rho({\cal O}) := \gamma^0 \, {\cal O}\, \gamma^0 \qquad \forall
  {\cal O}\in {\cal B}(L^2(\Man, S)).
\end{equation}
Indeed, for any $a=(f, g)$ in $\cinf\otimes \C^2$, one has from \eqref{FormOfa}
        \begin{align}
         \gamma^0  \,\pi(a)\,\gamma^0&=
  \begin{pmatrix}
  0&  \I_{2^{m-1}}  \\          
      \I_{2^{m-1}}  &0  \end{pmatrix}
 \begin{pmatrix}
    f\I_{2^{m-1}} & 0 \\ 0 &         
      g\I_{2^{m-1}}    \end{pmatrix}
\begin{pmatrix}
0 & \I_{2^{m-1}} \\ \I_{2^{m-1}} & 0       
    \end{pmatrix},\\
\label{eq:rhoprodlong}
&=  \begin{pmatrix}
    g\I_{2^{m-1}}& 0 \\ 0 &         
      f\I_{2^{m-1}}    \end{pmatrix}= \pi(\rho(a)).
        \end{align}
The unitary defining the automorphism \eqref{eq:rhoo}
induces an inner product 
	\begin{equation}
		\label{eq:rhoprod}
		 (\psi, \varphi):=\langle \psi, \gamma^0\varphi\rangle
		\quad\quad \forall \psi,\varphi\in L^2(\Man, S),
	\end{equation} 
with respect to whom the adjoint of any operator $\cal O$ in ${\cal B}(L^2(\Man,S))$ is 
	\begin{equation}
		\label{eq:oplus}
		{\cal O}^+ := \rho(\cal O)^\dag,
	\end{equation}
for 
\begin{align}
(\psi , {\cal O}\varphi)=   \langle \psi , \gamma^0{\cal O}\varphi\rangle=
 \langle {\cal O}^\dag \,\gamma^0 \psi ,\varphi\rangle
&= \langle \gamma^0\,\gamma^0  {\cal O}^\dag \gamma^0 \psi
  ,\varphi\rangle,\\
& = \langle \gamma^0{\cal O}^\dag \gamma^0 \psi ,\gamma^0\varphi\rangle=(\OO^+\psi ,\varphi).
\end{align}
 The product \eqref{eq:rhoprod}, called  \emph{twisted product}, is no longer definite positive: it coincides with the Krein product of spinors in
  lorentzian signature \cite{devastato2018lorentz}.

The adjoint \eqref{eq:oplus} is an
involution on ${\cal B}(L^2(\Man, S))$: \eqref{eq:rhoo} is a
$*$-automorphism (being inner) and 
$\rho^2$ is the identity, hence
\begin{align}
  \label{eq:invol}
&(\OO \OO')^+ = \rho(\OO\OO')^\dag =\rho(\OO')^\dag \rho(\OO)^\dag = {\OO'}^+
\OO^+,\\
 &(\OO^+)^+ = (\rho(\OO)^\dag)^+=\rho(\rho(\OO)^\dag)^\dag=\OO \qquad
\forall\OO, \OO'\in{\cal B}(L^2(\Man, S)).
\end{align}
Pulling \eqref{eq:oplus} back to the algebra yields a new involution
\begin{equation}
  a^+:= \rho(a)^* \quad \forall a\in\cinf\otimes\C^2.
\end{equation}
It is compatible with the representation  $\pi$ since \eqref{eq:rhoprodlong} guarantees that
\begin{equation}
  \pi(a^+)= \pi(\rho(a)^*) = \pi(\rho(a))^\dag = \rho(\pi(a))^\dag=\pi(a)^+.
\end{equation}
So one can safely remove the symbol of representation and use
without ambiguity $a^+$ to denote either the element of the algebra, or its
representation. As well $\rho(a)$ equivalently means the twisting
automorphism $\rho$ applied to $a\in\A$, or the
inner automorphism \eqref{eq:rhoo} applied to $\pi(a)$. 
Beware:
this does not mean  that the twisting automorphism is an inner
automorphism of $\cinf\otimes\C^2$ (as improperly suggested  in \cite{devastato2018lorentz}), for there is no
element of $\cinf\otimes\C^2$ whose representation is $\gamma^0$.
\newpage

	An  operator ${\cal O}\in{{\cal B}}(L^2(\Man, S))$ unitary with respect
          to the $\rho$-product \eqref{eq:rhoprod}, 
          \begin{equation}
\label{eq:OOounit} {\cal O}^+ {\cal O}={\cal O} {\cal O}^+ =\mathbb I,
\end{equation}
is said \emph{$\rho$-unitary} (or \!\emph{twisted-unitary}). Pulling
this property back to the algebra yields the following definition,
where $\bf 1$ denotes the unit of $\cinf\otimes\C^2$.
	\begin{definition}
		\label{def:rhounit}
	 The $\rho$-unitaries of $\cinf\otimes\C^2$ is the set
                \begin{align}
{\cal U}_\rho&:=\left\{u_\rho\in \cinf\otimes\C^2\quad\text{ such
that }\quad u_\rho^+u_\rho = u_\rho u_\rho^+ =\bf 1\right\}.
\end{align}
\end{definition}

\noindent 		From \eqref{FormOfa} and \eqref{eq:minimantwist},
  $a=(f,g)$ in $\cinf\otimes\C^2$ is
                $\rho$-unitary if, and only if, $\bar g=\frac 1{f}$. 
Thus ${\cal U}_\rho$ is isomorphic to the (multiplicative) group
$C^\infty_*(\Man)$ of
smooth functions on $\Man$ that never vanish. 
	\begin{remark}
\label{rem:ribbon}
 Unitarity and $\rho$-unitarity are not mutually exclusive:  for $f=\exp(i\theta)$, $g=\exp(-i\theta))$ with
		$\theta$ a real function, then $a=(f, g)$ is both
                unitary and $\rho$-unitary. Another example of
                unitary, $\rho$-unitary operators are the rotations,
                see \S\ref{sec:lorentz}	\end{remark}

 It is well known that if $u$ is a unitary element of the algebra $\A$
of a real spectral triple, then such
is $\text{Ad}(u)$ \cite[Lemma
5.1]{TwistGaugeLandiMarti2018}. The same is true for
the $\rho$-unitaries of a minimally twisted even dimensional manifold.
To see it, one first notices that $\gamma^0$ anticommutes with $\cal
J$ by \eqref{eq:kodim5},  so the inner
automorphism \eqref{eq:rhoo} is compatible with the real structure, in that 
\begin{equation}
  \label{eq:comp}
\rho(\J \OO\J^{-1})=\J \rho(\OO)\J^{-1}\qquad \forall \OO\in{\cal B}(\HH).
\end{equation}
In particular, this means that
\begin{equation}
\label{eq:rhocomp}
  \rho(\text{Ad}(a))=\text{Ad}(\rho(a)) \qquad \forall a\in\cinf\otimes\C^2.
\end{equation}

   \begin{proposition}
 \label{prop:antimorf}
For any $u_\rho\in {\cal U}_\rho$, one has that $\text{Ad}(u_\rho)
=u_\rho \J u_\rho \J^{-1}$ is $\rho$-unitary.  \end{proposition}
\begin{proof}
 For any $a$ in $\cinf\otimes\C^2$, using the order zero condition
 \eqref{eq:orderzero} one has
  \begin{equation}
    \text{Ad}(a) ^\dag = (a\J a\J^{-1})^\dag = \J a^\dag \J^{-1}a^\dag
    =a^\dag \J a^\dag \J^{-1}=\text{Ad}(a^*).
  \end{equation}
Together with \eqref{eq:rhocomp}, this yields
\begin{equation}
  \text{Ad}(a)^+=\rho(\text{Ad}(a))^\dag =
  \text{Ad}(\rho(a))^\dag=\text{Ad}(\rho(a)^*)
=\text{Ad}(a^+).
\end{equation}
Hence, again by the order zero condition, for any $u_\rho\in{\cal
  U}_\rho$ one has
\begin{align*}
  \text{Ad}(u_\rho)^+\text{Ad}(u_\rho)&=\text{Ad}(u_\rho^+)\text{Ad}(u_\rho),\\
&=
  u_\rho^+ \J u_\rho^+\J^{-1} u_\rho \J u_\rho \J^{-1}
=\J u_\rho^+\J^{-1}  \J u_\rho \J^{-1}=\I
\end{align*}
and similarly for  $\text{Ad}(u_\rho)\, \text{Ad}(u_\rho)^+$. \end{proof}
 
\begin{remark} 
\label{rem:unitaries}
The first Dirac matrix is not the only unitary matrix $R$ that implements
the automorphism $\rho$ on $\pi(\A)$, that is such that $R\pi(a)R^\dag
= \pi(\rho(a))$. Any such $R$ defines a twisted product
\begin{equation}
\label{eq:twistpRod}  (\psi,\varphi)_R:= \langle \psi, R\varphi\rangle.
\end{equation}
All these products yield the same involution $^+$ on $\A$
\cite{Martinetti:2024aa} (but not on ${\cal B}(\HH)$), and proposition \ref{prop:antimorf} does not depend on this
choice as soon as 
the compatibility with the real structure \eqref{eq:comp} holds. The freedom in
  the choice of $R$ is relevant for the fermionic
action, as investigated below in 
   proposition \ref{PropGenTwist}.{\footnote{Thank to F. Besnard for
       noticing that.}}
 \end{remark}

		\subsection{Torsion by group action}
\label{subsec:flucgroupact}

Given a real twisted spectral triple $(\A, \HH, D)$ with automorphism
$\rho$  compatible with the real structure in the sense of
\eqref{eq:rhocomp} and such that $\rho^2=\I$
(conditions all satisfied by the minimal twist of a manifold), one
has 
\begin{align}
  \text{Ad}(\rho(u))^+&=  \rho(\text{Ad}(\rho(u)))^\dag=\text{Ad}(u)^\dag=\text{Ad}(u)^{-1}\quad   \forall u\in{\cal U}({\cal A}),
       \end{align}
where the last equality follows from $\text{Ad}(u)$ unitary. 
  Therefore   \eqref{eq:conjugactt} becomes
   \begin{equation}
  \label{eq:newlook}   D_{A_\rho^u}= \text{Ad}(v) D_{A_\rho}
  \text{Ad}(v)^+ \quad\text{ for} \quad v=\rho(u).
   \end{equation}
\noindent A twisted gauge
transformation is 
thus the conjugate action - with respect to the~twis\-ted involution $+$ - of
the operator  $\text{Ad}(v)$ for $v$ unitary ($v$ is unitary since
$\rho$ is a $*$-automorphism, as a consequence of the regularity
condition together with the hypothesis $\rho^2=\I$). 

In a symmetric way, one may be interested in the conjugate action -
with respect to the initial involution $*$ - of $\text{Ad}(u_\rho)$
for $ u_\rho$  a $\rho$-unitary, namely
\begin{equation}
\label{eq:twistact}
  D\mapsto \Ad(u_\rho) \, D \Ad(u_\rho)^\dagger \quad \text{ for }
  u_\rho\in\cal U_\rho(\A).
\end{equation}
For the minimal
twist of a manifold, as shown in proposition
\ref{prop:gentorsion} below, this action generates the torsion term.

Let us first investigate the general form of \eqref{eq:twistact}.
\begin{lemma} 
\label{prop:urhogen} For any $u_\rho\in\ura$ one has 
          \begin{equation}
\label{eq:action3}
            Ad(u_\rho)\,\Dir\,
			Ad(u_\rho)^\dagger=   D + {A_\rho} +\epsilon' \J
                        A_\rho \J^{-1}\quad \text{
                          with } A_\rho=u_\rho[\Dir, u_\rho^*]_\rho.
          \end{equation}
	\end{lemma}
        \begin{proof} 
Following \cite{chamseddine2013inner}, let us denote
\begin{equation}
\hat u_\rho = \J u_\rho \J^{-1},
\end{equation}
so that $\text{Ad}(u_\rho)=\hat u_\rho
u_\rho$. 
Therefore
			 \begin{align}
	Ad(u_\rho)\,\Dir\, Ad(u_\rho)^\dag&= \hat u_\rho(u_\rho\Dir u_\rho^\dag)\hat
                                  u_\rho^\dag=\hat u_\rho(u_\rho u_\rho^+\Dir+u_\rho[\Dir,
                                  u_\rho^*]_\rho)\hat u_\rho^\dag,\\
			&=\hat u_\rho \Dir \hat u_\rho^\dag+\hat
                          u_\rho u_\rho \hat u_\rho^+ [\Dir,
                          u_\rho^*]_\rho\nonumber,\\
\label{eq:interm1}
           	&=\hat u_\rho \hat u_\rho^+ D+\hat u_\rho[\Dir, \hat u_\rho^\dag]_\rho+
                          u_\rho [\Dir,
                          u_\rho^*]_\rho
						\end{align}
where in the first line we use  $\rho(u_\rho^*)=\rho(u_\rho)^\dag = u_\rho^+$, in the second line we apply the twisted
first-order condition \eqref{eq:twistfirstord},  written as
\begin{equation}
  \label{eq:twis1use}
[D, u_\rho^*]_\rho \, \hat u_\rho^* = \rho(\hat u_\rho^*) D
u_\rho^*=\hat u_\rho^+ D u_\rho^* ,
\end{equation}
and in the third line we use
  $\rho(\hat u_\rho^\dag)=\rho(\hat u_\rho)^\dag=\hat u_\rho^+$. 
The result follows noticing that
\begin{align}
  \hat u_\rho^\dag = \J u_\rho^*\J^{-1},\quad \rho(\hat u_\rho^\dag) = \J \rho(u_\rho^*)\J^{-1}
\end{align}
so that
\begin{equation}
 \label{eq:JDUJ} 
[D, \hat u_\rho^\dag]_\rho = D\, \J u_\rho^*J^{-1} -  \J
  \rho(u_\rho^*)\J^{-1} \, D = \epsilon' \J [D, u_\rho^*] \J^{-1} .
\end{equation}
                                      


        \end{proof}

Applied to the Dirac operator $\ds$ of the minimal twist of a manifold, the action
\eqref{eq:twistact} generates a twisted fluctuation of the metric.
                \begin{proposition}
\label{prop:gentorsion}                  
             In $KO$-dimension $0$ and $4$,  the conjugate action on
             $\ds$ of the twisted
                  unitary
               $u_h:= (h, \frac 1{\bar h})$
with $h\in C^\infty_*(\Man)$
generates the additional term
                \eqref{eq:twisterm} with 
                \begin{equation}
                \omega_f = d(\ln |h|^2).
              \end{equation}
                \end{proposition}
                \begin{proof}
                  The expression 
                  \eqref{eq:rhocommutator} of the twisted commutator, together with \eqref{eq:kodim4} yield
                  \begin{align}
\label{eq:udu}
                    u_h [D, u_h^*]_\rho&=
                    -i\gamma^\mu\begin{pmatrix}
                      \frac 1{\bar h} & 0\\ 0& h
                    \end{pmatrix}\begin{pmatrix}
                                               \partial_\mu \bar h        &0  \\ 0&  \partial_\mu\frac 1h
                                                                                            \end{pmatrix}
=   -i\gamma^\mu\begin{pmatrix}
                                                       \frac 1{\bar
                                                         h}\partial_\mu
                                                       \bar h&0\\
0 & h\partial_\mu  \frac 1h 
                                                     \end{pmatrix}
                  \end{align}
(we omit $\I_{2^{m-1}}$ in the matrix).
Then by \eqref{eq:kodim40} and \eqref{eq:JDUJ} one gets
\begin{align*}
\J u_\rho[D,u_\rho^* ]_\rho \J^{-1}= \J u_\rho \J^{-1}\, \J [D,u_\rho^*
  ]_\rho \J^{-1} = u_\rho^* [D, \hat u_\rho^\dag]_\rho=-i\gamma^\mu\begin{pmatrix}
                                                       \frac 1{
                                                         h}\partial_\mu h&0\\
0 & \bar h\partial_\mu  \frac 1{\bar h} .
                                                     \end{pmatrix}
\end{align*}
Summing up with \eqref{eq:udu}, one obtains from \eqref{eq:action3}
\begin{equation}
\label{eq:addsad}  \text{Ad}(u_h)\, \ds\, \text{Ad}(u_h)^\dagger = \ds -i\gamma^\mu\begin{pmatrix}
                                                      \partial_\mu \ln
                                                      |h|^2&0\\
0 & -\partial_\mu \ln |h|^2\end{pmatrix}=\ds -i\gamma^\mu \partial_\mu \left(\ln
                                                      |h|^2\right)\gamma
\end{equation}
where we use  
\begin{equation}
  \frac 1h \partial_\mu h +   \frac 1{\bar h} \partial_\mu \bar h =
  \frac{\bar h \partial_\mu h + h\partial \bar h}{\bar h h}=
  \frac{\partial(\bar h h)}{|h|^2}= \frac{2|h|\partial_\mu
    |h|}{|h|^2}=2 \partial_\mu(\ln |h|)= \partial_\mu (\ln |h|^2).
\end{equation}
The second term on the diagonal follows from the Leibniz rule
\begin{equation}
  h\partial_\mu \frac 1h = \partial_\mu (\frac hh)- \frac
  1h \partial_\mu h = - \frac
  1h \partial_\mu h
\end{equation}
and similarly for the complex conjugate.                \end{proof}

		\begin{corollary}
\label{prop:f+fprime}
	In $KO$-dimension $0, 4$, the conjugate action on the twisted covariant Dirac operator
$\ds_{\omega_f}$  \eqref{eq:domegaf} of the twisted
                  unitary
    $u_{h'}:= (h', \frac 1{\bar{h'}})$
with $h'\in C^\infty_*(\Man)$ amounts to mapping $\omega_f$ to 
\begin{equation}
  \omega_f + d(\ln|h'|^2).
\end{equation}
	\end{corollary}
		\begin{proof}
The additional term $-i\gamma^\mu f_\mu$ is invariant under the
considered group action: using the
notations of lemma \ref{prop:urhogen} and remembering that any capped quantity commutes with
non capped ones by the order zero condition, one has 
			\begin{align}
				\Ad(u_h)(-i\gamma^\mu f_\mu\gamma)
                          \,\Ad(u_h)^\dagger&= -if_\mu\left(\hat
                          u_h  u_h\,\gamma^\mu \, u_h^\dagger\hat
                          u_h^\dagger\right) \gamma,\\
&= -if_\mu\left( \hat u_h\, \gamma^\mu \,\hat  u_h^\dagger\right) \gamma=-if_\mu\gamma^\mu\gamma,
			\end{align}  
		where we first use that in $KO$-dimension $0, 4$ the
                grading $\gamma$ not only commutes with $u_h^\dag$ by
                definition, but also with $\hat u_h^\dag$, since
                $\gamma$  commutes with
                $\J$; then we apply  \eqref{eq:kodim4} to $u_h^\dag$,
                then to $\hat
                u_h^\dag$. Using  \eqref{eq:addsad} one 
                finally obtains
			\begin{align}
				\Ad(u_{h'})\, \ds_{\omega_f}\, \Ad(u_{h'})^\dagger&=
                          \Ad(u_{h'})\, \ds\,\Ad(u_{h'})^\dagger
                          -if_\mu\gamma^\mu\gamma,\\
\nonumber
&=\ds - i\gamma^\mu
                          \left(f_\mu + \partial_\mu \ln
                                                      |h'|^2\right)\gamma 
			\end{align}  
\end{proof}

In the definition of $\ds_{\omega_f}$ the $1$-form $\omega_f$ 
\eqref{eq:defomega}  is
arbitrary, it does not need to be exact. The conjugate action of twisted
unitaries adds to it an exact  $1$-form
$d(\ln|h|^2)$.
\newpage  Therefore not every torsion may be obtained from this
action.

\begin{proposition}
The 
conjugate action \eqref{eq:action3} of the group of twisted unitaries generates all the torsions whose
associated $3$-form is co-exact,
\begin{equation}
K^\flat = \delta(f\nu_g)
\end{equation}
where $\delta=-\star d\star$ is the co-derivative and $\nu_g$ is the
volume form of $\Man$.
\end{proposition}
\begin{proof}
 One generates a twisted
fluctuation with torsion $K^\flat=-\star df$ for an arbitrary  $f\in\cinf$ by choosing $h=e^{\frac
  f2}$ in proposition \ref{prop:gentorsion}. The result then follows
remembering \cite{Lee:2010aa}  that for a $0$-form
one has $\star\star
f=f$  and $\star f =f\nu_g$, so that
\begin{equation*}
  \delta (f\nu_g)= -\star d \star (f\nu_g) =- \star d(\star\star f) =-
  \star df.
\end{equation*}

\vspace{-.35truecm}\end{proof}

The action \eqref{eq:twistact} preserves
    the selfadjointness of $D$, in agreement with
    the torsion being a selfadjoint fluctuation.
A gauge transformation \eqref{eq:newlook}  preserves $\rho$-adjointness (which could be
relevant in case one starts with a $\rho$-adjoint operator $D$ \cite{devastato2018lorentz,Nieuviarts:2024aa}) but
not necessarily selfadjointness. This is not a problem
here since $\text{Ad}(u)$ is
trivial by \eqref{eq:adunit}, but it becomes important for the Standard Model or in electrodynamics: in \cite{martinetti2022lorentzian} and
\cite{filaci2021minimal}  we restrict to gauge
transformations that  preserve selfadjointness (they contain, but
do not reduce to ${\cal U}(\A)\, \bigcap \,{\cal U}_\rho(\A)$
\cite[Remark 5.9]{martinetti2022lorentzian}) and left
as an open question 
non-selfadjoint twisted
gauge transformations of selfadjoint operators $D$. 

The actions
\eqref{eq:newlook} and \eqref{eq:twistact} are two
  symmetric ways to entangle the involutions: one considers the conjugate action - with respect
    to the one - of a unitary with
    respect to the other. The study of non-entangled action follows
    from the following 
    \begin{lemma} 
\label{prop:nonent}
In $KO$-dimension $0, 4$, for
      $a=(f,g)\in\cinf\otimes\C^2$,
      \begin{equation*}
 \left\{
 \begin{array}{l}
   Ad(a)\,\ds\, (Ad(a))^+ \\[4pt]
Ad(a)\,\ds\, 
      (Ad(a))^\dag
\end{array}\right. \text{ is a twisted
      fluctuation iff }
\left\{
 \begin{array}{l}
  a \text{ is unitary},\\[4pt]
a=u u_\rho \text{ with } u\in{\cal U}(\A), u_\rho\in{\cal U}_\rho(\A).
\end{array}\right.
    \end{equation*}
    \end{lemma}
    \begin{proof}
By repeating the calculation of lemma
      \ref{prop:urhogen} and, one
      obtains 
      \begin{align}
        \label{eq:T1}
        Ad(a)\Dir (Ad(a))^+&=\hat a \hat a^\dagger
                             aa^*\Dir+\hat a
                             \hat a^\dagger a[\Dir,
                             a^+]_\rho+\epsilon^\prime aa^\dagger \,\J
                             a[\Dir, a^+]_\rho \J^{-1},\\
        \label{eq:T2}
        Ad(a)\,\Dir\, Ad(a)^\dagger&=\hat a \hat  a^+ aa^+\Dir +  \hat a \hat a^+ a[\Dir,  a^*]_\rho
                                     +\epsilon^\prime aa^+\, \J a[\Dir, a^*]_\rho \, \J^{-1}.
      \end{align}
      To be of the form $D + A + JAJ^{-1}$, one needs the term in
      front of $D$ to be the identity. For 
      \eqref{eq:T1},  noticing that $\hat a\hat a^\dagger =  \J a a^*
      \J^{-1}=a a^*$ by \eqref{eq:kodim40}, this means $b^2=\I$
for      \begin{equation}
\label{eq:cond1}
b:=
      aa^\dagger =
      \begin{pmatrix}
        f\bar f & 0 \\ 0& g\bar g
      \end{pmatrix}, \quad \text{ that is }|f|=|g|=1. 
    \end{equation}
Hence $a=(e^{i\theta}, e^{i\varphi})$ for some
  $\theta, \varphi\in \cinf$ is unitary.

For \eqref{eq:T2}, noticing that $\hat a \hat a^+ = \J
aa^+\J^{-1}=(aa^+)^*$, one obtains $c^*c=\I$ for 
\begin{equation}
 c:= aa^+ =
  \begin{pmatrix}
  f\bar g & 0 \\ 0&\bar fg  
  \end{pmatrix}, \quad \text{ that is } |fg|=1.
\end{equation}
So $a=(re^{i\theta}, r^{-1}e^{i\varphi})$
with $r, \theta,\varphi\in\cinf$ is the product of $(e^{i\theta},
e^{i\varphi})\in{\cal U}(\A)$ by $(r, r^{-1})\in{\cal U}_\rho(\A)$.
\end{proof}

By \eqref{eq:adunit},  $\text{Ad}(u u_\rho)=\text{Ad}(u)\text{Ad}(u_\rho)$
  reduces to $\text{Ad}(u_\rho)$. The proposition above then shows
  that non-entangled actions - that is the conjugate action with respect
  to an involution of a unitary for the same involution - do not
  generate twisted fluctuations (except if the operator
  is both unitary and $\rho$-unitary).

\newpage
	
	\section{Action formulas and change of signature}
		\label{sec:actions}

The action for a spectral triple is the sum of the
\emph{fermionic} and \emph{spectral} ones. For the
spectral triple of the Standard Model, the former describes the
coupling between fermions and bosons (including the Higgs),
the latter describes the self interactions of bosons (Yang-Mills
terms), the Higgs mass term and its quartic potential, as well as gravitational terms
including a minimal coupling with the Higgs.

The fermionic action has been adapted to twisted case in
\cite{devastato2018lorentz}, and studied in details for the
spectral triple of electrodynamics in \cite{martinetti2022lorentzian}.
 In that case, it turns out that the extra term generated by the twisted fluctuation
 yields the $0^\text{th}$ component of the momentum-energy $4$-vector
 in \emph{lorentzian signature}. In the light of the results of the
 previous section, this means that the torsion term arising from the
 minimal twist of a riemannian manifold gets interpreted, through the
 fermionic action, as energy-momentum in lorentzian signature. 
We study
 this interplay between torsion and change of signatures in
 \S\ref{subsec:torsionergy}, showing how this
limits the choice of the unitary $R$ to the
sole $\gamma^0$ matrix.

Besides generating
torsion as shown above, twisted unitaries
also implement Lorentz invariance for the
fermionic action. This is shown, for minimally twisted manifolds, in
\S \ref{sec:lorentz}.

Regarding the spectral action, some proposal for a twisted version have
been formulated in \cite{TwistSpontBreakDevastaMartine2017} and
\cite{devastato2018lorentz}. None of them is fully satisfactory, yet
we provide an explicit calculus of spectral action with torsion in \S \ref{spectralAct}. 

\subsection{Fermionic action} 
		\label{fermact}

        The fermionic action
        for a real twisted spectral triple $(\A, \HH, D)$, defined in
        \cite{devastato2018lorentz} as
        \begin{equation}
\label{eq:defactferm}
          	\act_R(\Dir_{A_\rho}):= {\frak A}_{D_{A_\rho}}^R(\tilde\psi, \tilde \psi) ,
        \end{equation}
is the evaluation -  for $D=D_{A_\rho}$ - of the bilinear form 
	 \begin{align}
			\label{FerAction}
			{\frak A}_D^R(\phi, \psi) :=
           ( J \phi, D\psi)_R\qquad\forall
           \phi, \psi\in\HH
		\end{align}
on the Gra{\ss}man vector	$\tilde \psi$,
associated with a vector $\psi$ in the $+1$ eigenspace of  the unitary $R$ that
implements the twist (having
in mind $R=\gamma^0$, it was implicitly assumed that $R$ were selfadjoint, hence with eigenvalues $\pm 1$).
	
This action is invariant \cite[Prop. 4.1]{devastato2018lorentz} under the twisted-gauge
 transformation \eqref{eq:newlook} of the Dirac operator combined with the action of
 unitaries on $\psi$ 
 \begin{equation}
   \psi\mapsto \text{Ad}(u)\psi \qquad \forall u\in {\cal U}(\A).
 \end{equation}

As stressed in remark \ref{rem:unitaries},  the unitary $R$ that
implements the twist is not
unique and the action depends on it through the twisted
product \eqref{eq:twistpRod} (that is why we changed the notations of
\cite{devastato2018lorentz,martinetti2022lorentzian} and use $R$ instead
of $\rho$ in \eqref{eq:defactferm}). In
particular, for minimally twisted manifolds (even dimensional), the flip \eqref{eq:minimantwist} is implementable by any odd
product of distinct euclidean $\gamma$ matrices
 \begin{align}
\label{eq:defR}
R=\prod_{i=1}^{k}\gamma^{a_i} \;\text{ with } \; k\leq 2m \;\text{ odd
    and } \;\gamma^{a_i}\neq\gamma^{a_j} \; \forall\; i, j=1, ...,k
\end{align} 
(one safely assumes that all the matrices are distinct,  for any pair $\gamma^{a_i}=\gamma^{a_j}$  cancels as
  $(\gamma^{a_i})^2=\mathbb I$ after some permutations).

\newpage
\begin{proposition}
\label{PropGenTwist}
The inner automorphism induced on ${\cal B}(L^2(\Man, S))$ by any
unitary $R$ \eqref{eq:defR} is an extension  \eqref{eq:rhoprodlong}
of the flip \eqref{eq:minimantwist}.
Moreover 
$R^\dag = (-1)^lR$ where $k=2l+1$.
\end{proposition}
\begin{proof}
$R$ is unitary because such are any single 
$\gamma^{a_i}$ on even dimensional manifolds. It
anticommutes with $\gamma$, for any 
$\gamma^{a_i}$ anticommutes with the $2k-1$ matrices $\gamma^a$, $a\neq a_i$,
in \eqref{eq:gamma5}. Therefore, \eqref{eq:twistcomgpi} yields
\begin{align}
  R\pi(a)R^\dag &= R\frac{\I -\gamma}2\pi_0(f)R^\dag +  R\frac{\I
  +\gamma}2\pi_0(g) R^\dag ,\\
&=  \frac{\I +\gamma}2 \pi_0(f) +
  \frac{\I-\gamma}2\pi_0(g)=\pi(\rho(a)) \quad \forall (f, g)\in
  \cinf\otimes \C^2, 
 \end{align}
The last statement is checked calculating
\begin{align*}
				R^\dagger&={\gamma^{a_{2l+1}}}^\dagger\dots
                                                {\gamma^{a_1}}^\dagger
={\gamma^{a_{2l+1}}}\dots {\gamma^{a_1}}=(-1)^{(2l+1)l} \gamma^{a_1}\dots
  \gamma^{a_k}=(-1)^lR.
			\end{align*}
\end{proof}
			
			In order to make sense when applied to Gra{\ss}man variables,  the bilinear form
                        \eqref{FerAction} is asked to be
                        antisymmetric \cite{chamseddine2007gravity}.  In case  $R=\gamma^0$
                        \cite{devastato2018lorentz, martinetti2022lorentzian}, this is
                        obtained by  taking $\psi$ in the $+1$
                        eigenspace of $\gamma^0$. But this is not the
                        only possibility. By the
                        previous lemma, any unitary $R$ \eqref{eq:defR} 
                        is either selfadjoint and has eigenvalues
                        $\pm 1$, or is skewadjoint with eingenvalues
                        $\pm i$. In both case we denote
                        \begin{equation}
\label{eq:alpha}
                          \HH_R^+:=\left\{\psi\in\HH, R\psi= \alpha\psi
                            \text{ where }\left\{  \begin{array}{ll}
\alpha=1 &\text{ in case $l$
                              is even}, \\
\alpha= i &\text{ in case  $l$
                              is odd.}
                            \end{array}\right.
\right\},
                                                    \end{equation}
and define $\HH_R^-$ in a similar way with $\alpha=-1, -i$ instead
of $1$, $i$. 	\begin{lemma}
				\label{skewDom}
			For any $D$ selfadjoint such that the real
                        structure $\J$ of the manifold satisfies \eqref{AxTripl}, and $R$ as in
                        \eqref{eq:defR}, 
                                \begin{equation}
                         	{\frak A}_D^\rho(\phi,
                                \psi)=\epsilon\epsilon''\bar\alpha^2\,{\frak A}_D^\rho(\psi, \phi)\qquad
                                \forall \psi, \phi\in\HH^+_R\,\text{
                                  or }\; \psi, \phi\in\HH_R^-.
                                \end{equation}
\end{lemma}
			\begin{proof}
			The proof is similar to 
                        \cite[Prop. 4.2]{devastato2018lorentz}, once
                        noticed that $\J$ is compatible with the twist
                        in the sense of \eqref{eq:comp}, for
                        \begin{equation}
\label{eq:rj}R\J=\epsilon' \J R, \quad  R^\dag \J = \epsilon' \J R^\dag
\end{equation}
(by \eqref{eq:kodim5} $\J$ anticommutes
                        with any odd product of distinct $\gamma$ matrices).  One has
				\begin{align}
				{\frak A}_D^R(\phi, \psi)
                                  &=\langle \J\phi,R D\psi\rangle
                                    =\epsilon \langle \J \phi, J^2R
                                    D\psi\rangle=\epsilon \langle \J
                                    RD\psi, \phi\rangle,\\
&=\epsilon {\epsilon'}^2\langle R  D  \J \psi,\phi \rangle=\epsilon \langle     \J \psi,DR^\dagger\phi \rangle\\
					&=\bar\alpha\epsilon
                                         \langle   \J R^\dagger R  \psi,D\phi
                                          \rangle=\bar\alpha \epsilon
                                          \epsilon'J\langle
                                          R^\dagger \J R \psi,D\phi \rangle=\bar\alpha\epsilon \epsilon^{\prime}\langle     \J R \psi,RD\phi \rangle\\
					&=\bar\alpha^2\epsilon
                                          \epsilon^{\prime}\langle
                                          \J  \psi, RD\phi \rangle= \bar\alpha^2\epsilon\epsilon'	{\frak A}_D^\rho(\psi, \phi)
\label{eq:oddgamma}			
	\end{align}
		where the first line follows from \eqref{AxTripl} and $J$ being
                antiunitary (meaning $\langle J\phi, J\psi\rangle =
                \langle \psi, \phi\rangle$, this was miswritten  \cite[Prop. 4.2]{devastato2018lorentz}), the second line follows from \eqref{AxTripl} , the third
                from $R^\dag \psi=\bar\alpha\psi$ then again \eqref{AxTripl},   the
                last line is obtained from $R\psi=\alpha\psi$.
\end{proof}

Nonzero  selfadjoint twisted fluctuations occur in
$KO$-dimension $0$, where $\epsilon
\epsilon'=1$,  and $KO$-dimension $4$, where $\epsilon
\epsilon'=-1$. We thus conclude that in the first case,  $\frak A$ is antisymmetric only for $l$
odd, in the second case for $l$ even. 

\newpage

\subsection{Torsion as energy-momentum}
\label{subsec:torsionergy}

On a $4$-dimensional manifold $\Man$ - which is the case of interest
for the Standard Model and  the dimension in which the
interpretation of the twisted fluctuation as a torsion via
corollary \ref{cor:flucttors} is possible - there are  two
odd numbers $k=2l+1$~smaller~than the dimension: $k=1$ (that is $l=0$)
or $3$ ($l=1$). The
$KO$-dimension of a minimally twisted manifold coincides with its
metric dimension, so by the remark of the preceding paragraph there
remains  only
$l=0$, that is $R=\gamma^a$ a single Dirac matrix.

In \cite{martinetti2022lorentzian}, by  comparing the fermionic action
on a twisted \emph{riemannian} manifold for  $R=\gamma^0$ 
with the Weyl action on a  \emph{lorentzian} manifold (both of
dimension $4$), that is
\begin{equation}
\label{eq:weyl}
  i\Psi^\dagger(\partial_0 \pm \sum_{j=1}^3 \sigma_j\partial_j)\Psi
\end{equation}
(the sign depends on wether $\Psi$ is the right or left
handed component of a Dirac spinor), one sees that - up to a doubling
of the manifold discussed in remark \ref{rem:double} below - a plane wave solution of the twisted fermionic action coincides with a solution of
the Weyl equation with  energy the component $f_0$ in
\eqref{eq:domegaf}.
We show below that this interpretation of a  riemannian torsion as 
a lorentzian energy-momentum only occurs for  
$R=\gamma^0$. Other choices for $R$ induce no change of signature.

To see that, we calculate the fermionic action for $R=\gamma^a$ an
arbitrary euclidean Dirac matrix. A  Dirac spinor  $\phi=(\varphi_1, \varphi_2)$ satisfies $R\phi=\alpha\phi$ if and
only if $\alpha \varphi_1 = \sigma^a \varphi_2$ and $\alpha \varphi_2 =
\tilde\sigma^a \varphi_1$. Since $\tilde\sigma^a\sigma^a=\I$, this is
equivalent to 
\begin{equation}
\label{eq:eigenphi}\phi=
\begin{pmatrix}
  \varphi\\\alpha^{-1}\tilde{\sigma}^a\varphi
\end{pmatrix}
\text{ with } \varphi \text{ a Weyl
spinor and } \alpha=\pm 1.
\end{equation}
\begin{lemma}
For $\Man$ of dimension $4$,  $R=\gamma^a$ 
and $\psi, \phi$ in the same eigenspace of $R$,
 \begin{align}
  \label{eq:feract1}
\frak A^R_{\ds_{\omega_f}}= i\alpha\langle \J\phi, \gamma^\mu\omega_\mu
   \psi\rangle &+  \int_{\Man} \, \!{}^T\! \varphi \left( D^{\mu a}  \partial_\mu
 -F^{\mu a}  f_\mu\right)
\zeta\; d\nu_g 
\end{align}
\end{lemma}
where $\varphi$, $\zeta$ are the components of $\phi, \psi$ in
\eqref{eq:eigenphi}, $\J=i\gamma^0\gamma^2cc$ is the real structure
(withh $cc$ the complex conjugation) and one denotes
\begin{align}
  D^{\mu a}:=\sigma^2\sigma^\mu
    \tilde\sigma^a -\; {}^T{\tilde\sigma^a} \sigma^2 
  \,\tilde\sigma^\mu,\quad
F^{\mu a}:=\sigma^2\sigma^\mu
    \tilde\sigma^a +\; {}^T{\tilde\sigma^a} \sigma^2 
  \,\tilde\sigma^\mu\, .
\end{align}

\begin{proof}
By \eqref{eq:rj} one has 
			 	\begin{align}
\label{eq:relcalss}
			 		{\frak A}^\rho_D(\phi, \psi) &=\langle \J \phi,R D\psi\rangle=\langle R^\dagger \J \phi, D\psi\rangle=-\alpha\langle  \J  \phi, D\psi\rangle.
			 	\end{align}
On the one side, 
\begin{align*}
\J\phi =i\gamma^0\gamma^2 \circ cc\begin{pmatrix}
\varphi  \\
\alpha^{-1}\tilde{\sigma}^a\varphi  \\
\end{pmatrix}=i\begin{pmatrix}
\tilde{\sigma}^2 \bar\varphi \\
\sigma^2\bar\alpha^{-1}\overline{\tilde{\sigma}^a}\bar\varphi  \\
\end{pmatrix}.
\end{align*}
On the other side, denoting $\omega_\mu:= \widetilde
		\Gamma^b_{\mu a}\gamma^a\gamma_b$, one has
\begin{align}
&\ds_{\omega_f}\psi=-i\gamma^\mu(\partial_\mu + \omega_\mu +
  f_\mu\gamma)\psi=-i\gamma^\mu\omega_\mu\psi 
  -i \begin{pmatrix}
\alpha^{-1}\sigma^\mu \tilde{\sigma^a}(\partial_\mu- f_\mu)\zeta
 \\
\tilde{\sigma}^\mu(\partial_\mu+ f_\mu)\zeta
\end{pmatrix}.
\end{align}
Therefore, using $({\bar\sigma}^2)^\dag=(-i{\sigma_2})^\dag
=-i\sigma_2=\sigma^2$ and $(\sigma^2)^\dag=-\sigma^2$, one obtains
\begin{small}
  \begin{align*}
    \langle \J \phi, \ds_{\omega_f}\psi \rangle
    &=-i\langle \J\phi, \gamma^\mu\omega_\mu \psi\rangle -\alpha^{-1}\int_{\Man} \!\!\left({}^T\!\varphi \;\sigma^2\sigma^\mu
      \tilde\sigma^a(\partial_\mu-f_\mu){}\zeta\right)
      -\left( {}^T\!\varphi\; {}^T{\tilde\sigma^a} \,\sigma^2 
      \,\tilde\sigma^\mu(\partial_\mu +f_\mu)\zeta\right)\; d\nu_g .
  \end{align*}
\end{small}
The result then follows from \eqref{eq:relcalss}.
\end{proof}

\newpage

For $R=\gamma^0$,  the identification of  torsion as
energy-momentum is due to the disappearance of
$\frac{\partial}{\partial x_0}$ into the fermionic action, and the
appearance of the $f_0$ component of
the twisted fluctuation.
A similar result holds for $R=\gamma^a$ an arbitrary Dirac matrix.
\begin{proposition}
\label{prop:fermactfinal}
On a minimally twisted $4$-dimensional orientable, closed, riemannian manifold~$\Man$, the twisted
fermionic action is 
 \begin{align*}
\act_R(\ds_{\omega_f})&= i{\alpha}\langle J\tilde\psi, \gamma^\mu\omega_\mu
   \tilde\psi\rangle  \\
&+ \left\{
  \begin{array}{ll}
2\int_\Man  {}^T\!\tilde\zeta \,\sigma_2
  \left( if_0- \sum_{j\neq0}{\sigma}_j
  \partial_j \right)  \tilde\zeta d\nu_g& \text{for }  R=\gamma^0;\\[6pt]
 2\int_\Man  {}^T\!\tilde\zeta \,\sigma_2\sigma_a
  \left(\partial_0 +i\sum_{j\neq a} \sigma_j\partial_j+
    i\sigma_j f_j \right) \tilde\zeta
  d\nu_g&  \text{for } R=\gamma^a\neq\gamma^0.
  \end{array}\right.
\end{align*}
\end{proposition}
\begin{proof}
Since ${}^T\!\tilde\sigma^0=\tilde\sigma^0$ 
commutes with any $\sigma^\mu$, one has
\begin{align*}
  D^{\mu0} \partial_\mu&= \sigma^2\tilde\sigma^0(\sigma^\mu-\tilde\sigma^\mu)\partial_\mu=
  -2\sigma^2\tilde\sigma^0\sum_{\mu\neq
  0}\tilde\sigma^\mu\partial_\mu;\\
F^{\mu 0}f_\mu &= \sigma^2\tilde\sigma^0(\sigma^\mu+\tilde\sigma^\mu)f_\mu =
  2\sigma^2\tilde\sigma^0\tilde\sigma^0 f_0.
\end{align*}
Since ${}^T\!\tilde\sigma^2=-\tilde\sigma^2$ commutes with $\tilde\sigma^\mu$ 
for $\mu=0, 2$, anticommutes for $\mu=1,3$,
one~has
 \begin{align*}
  D^{\mu2}\partial_\mu&= \left(\sigma^2\sigma^\mu \tilde\sigma^2
     +\tilde\sigma^2\sigma^2\tilde\sigma^\mu\right)\partial_\mu=
                        \sigma^2\tilde\sigma^2 \left(
                        \sum_{\mu=0 ,2}(\sigma^\mu +\tilde\sigma^\mu)+ \sum_{\mu=1,3}
     (\tilde\sigma^\mu- \sigma^\mu)\right)\partial_\mu\\
&=
                        2\sigma^2\tilde\sigma^2\sum_{\mu\neq 2}\tilde\sigma^\mu\partial_\mu;\\
F^{\mu 2} f_\mu&=  \left(\sigma^2\sigma^\mu \tilde\sigma^2-
     \tilde\sigma^2\sigma^2\tilde\sigma^\mu\right)f_\mu=  
 \sigma^2\tilde\sigma^2\left(\sum_{\mu=0,2}(\sigma^\mu-\tilde\sigma^\mu)-\sum_{\mu=1,3}(\sigma^\mu+\tilde\sigma^\mu)\right)f_\mu
   \\
&
=-2\sigma^2\tilde\sigma^2\,\tilde\sigma^2  f_2.
  \end{align*}
Finally, since  ${}^T\!\tilde\sigma^a=\tilde\sigma^a$ for $a=1,3$ commutes with
$\tilde\sigma^\mu$ for $\mu=0, a$, anticommutes for $\mu=2, b$ (with
$b=1$ if $a=3$ and vice-versa), one gets
 \begin{align*}
  D^{\mu a}\partial_\mu&= \left(\sigma^2\sigma^\mu \tilde\sigma^a-
     \tilde\sigma^a\sigma^2\tilde\sigma^\mu\right)\partial_\mu=\sigma^2\tilde\sigma^a\left(\sum_{\mu=0,a}  (\sigma^\mu
   +  \tilde\sigma^\mu) + \sum_{\mu= 2, b}(\tilde\sigma^\mu - \sigma^\mu)\right)\\ &=
                        2\sigma^2 \tilde\sigma^a \sum_{\mu\neq a}\tilde\sigma^\mu\partial_\mu;\\
F^{\mu a} f_\mu&=  \left(\sigma^2\sigma^\mu \tilde\sigma^a+
     \tilde\sigma^a\sigma^2\tilde\sigma^\mu\right)f_\mu=  \sigma^2
                 \tilde\sigma^a\left(\sum_{\mu=0,a} (\sigma^\mu
                 -\tilde\sigma^\mu) - \sum_{\mu=2, b}(\sigma^\mu+
                 \tilde\sigma^\mu)\right) f_\mu,\\
&=                   -     2\sigma^2\tilde\sigma^a\, \tilde\sigma^a f_a.
  \end{align*}
The results then  follows from \eqref{eq:feract1} together with \eqref{eq:defsigmamu}.
\end{proof}

For $R=\gamma^0$, one retrieves the result of
\cite[Prop~3.5]{martinetti2022lorentzian}. Namely, on
\begin{equation}
\zeta(x_0, {\bf x})=e^{\pm if_0x_0}\xi({\bf x}) \quad\text{ with } \quad \xi({\bf x})=
\xi(x_1, x_2, x_3),
\end{equation}
 the operator $(if_0 -
\sum_{j\neq 0}\sigma_j \partial_j)$ in \ref{prop:fermactfinal} coincides
with the operator \linebreak $(\partial_0\pm\sum_{j=1}^3\sigma_j\partial_j)$ in the Weyl
action \eqref{eq:weyl}. Hence, modulo a doubling of the algebra
discussed below,  a plane wave solution of the equation
of motion obtained from the twisted fermionic action on a riemannian
manifold coincides with a
plane wave solution of the Weyl equation in lorentzian
signature.\newpage

In case $R=\gamma^a\neq \gamma^0$  is another Dirac matrix,
then  the operator
\begin{equation}
\partial_0
+ i\sum_{j\neq a}\sigma_j\partial_j + i\sigma_af_a
\end{equation}
that appears in proposition \ref{prop:fermactfinal},  applied on
 \begin{equation}
\zeta(x_0, {\bf x})=e^{f_a x_a}\xi(x_0, x_{i\neq a})
\end{equation}
coincides with the operator
\begin{equation}
(\partial_0 +
i\sum_{j=1}^3\sigma_j\partial_j).
\end{equation}
The latter appears in the euclidean
Weyl action, obtained from \eqref{eq:weyl} substituting the Pauli
matrices $\sigma_j$ with
their lorentzian counterpart $i\sigma_j$. In other
terms, although for $R=\gamma^a\neq \gamma^0$ the derivative along
$x_a$ is replaced with the component $f_a$ of the twisted
fluctuation, this does not correspond to a change of signature. To
summarise: 
\begin{corollary}
\label{corr:rgamma0}
  In $KO$-dimension $4$, the only choice, for $R$, of an odd product of
  Dirac matrices that implements a change of signature is $R=\gamma^0$. 
\end{corollary}

\begin{remark}
  \label{rem:double}
In order to suitably identify the lagrangian density obtained from
the twisted fermionic action with the Weyl lagrangian,  one needs to
consider the minimal twist of a doubled manifold, the latter being the product of a
manifold by a two point space (see \cite[\S 4]{martinetti2022lorentzian}). This does not interfere with the
conclusion of corollary \ref{corr:rgamma0}.  
\end{remark}

 \subsection{Lorentz symmetry}
\label{sec:lorentz}

As noted in
\cite{martinetti2022lorentzian}, the fermionic action for the minimal twist
of a $4$ dimensional manifold $\Man$  is invariant
under the action of the (restricted) Lorentz group $SO^+(1,3)$ simultaneously on spinors and on
the twisted covariant Dirac operator:
 \begin{equation}
\label{eq:aclorentz}
\ds_{\omega_f}\longmapsto S[\Lambda] \, \ds_{\omega_f} \,S[\Lambda]^{-1},\qquad   \psi\longmapsto
S[\Lambda]\psi \quad \forall \psi\in L^2(\Man, S) \end{equation}
where $S[\Lambda]$ is the spin representation of $\Lambda=
\exp(t_{ab}\Lambda^{ab})\in SO^+(1,3)$,
with $t_{ab}\in\mathbb R$ and $\Lambda^{ab}$ the generators of the
Lorentz group ($a,b=0, 1, 2, 3$). Explicitly, 
\begin{equation}
\label{eq:spinlorentz}
  S[\Lambda]=\exp(\frac{i}{2}t_{ab}T^{ab})
\end{equation}
where the spin representation of the
generators of $SO^+(1,3) $ are the commutator
\begin{equation}
\label{eq:genT}
T^{ab}\defeq -\frac i4\,  [\gamma^a_L, \gamma^b_L]
\end{equation}
of the Lorentzian Dirac matrices
\begin{equation}
\label{eq:lorgam}
 \gamma^0_L=\gamma^0\qquad \qquad\text{and}\qquad\qquad \gamma^j_L=i\gamma^j\qquad\text{for}\qquad j\in\{1, 2, 3\}.
 \end{equation}

 \begin{remark}
   The action \eqref{eq:aclorentz} of the Lorentz group is the usual implementation of the
   relativistic invariance of the Dirac equation, and for the minimal
   twist of the spectral triple of electrodynamic, it allows to
   interpret not only the component $f_0$ of the twisted fluctuation
   as an energy, but also $f_j$, $j=1, 2, 3$ as the
   corresponding component of the lorentzian energy-momentum
   $4$-vector in a boosted frame \cite{martinetti2022lorentzian}.
 \end{remark}
\newpage
What was missed in \cite{martinetti2022lorentzian} is that
$S[\Lambda]$ actually is a $\rho$-unitary operator. To see it, we
begin with an easy relation between euclidean and
lorentzian Dirac matrices.
\begin{lemma} 
\label{lem:loentz}
Let $\gamma^b_L$ with $b=0, 1, 2, 3$  be the Lorentzian Dirac matrices
  \eqref{eq:lorgam}, and  $\gamma^a$ an euclidean Dirac matrix. Then  
  \begin{equation}
\label{eq:lorentz+adj}
    \gamma^a\, (\gamma_L^b)^\dag\, \gamma^a = \gamma_L^b \quad \forall
    b=0, 1, 2, 3
  \end{equation}
if, and only if, $a=0$.
\end{lemma}
\begin{proof}
  The euclidean Dirac matrices are selfadjoint, square to $\I$, and each of them
  anticommutes with the others (and commutes with itself), therefore
  for $b=1,2,3$
  \begin{equation}
    \label{eq:rhoadgamma1}
    \gamma^a \,(\gamma_L^b)^\dag\, \gamma^a =\left\{
      \begin{array}{ll}
        -i\gamma^a \gamma^b \gamma^a=-i\gamma^b=-\gamma^b_L&\text{ if }a= b,\\
        -i\gamma^a \gamma^b \gamma^a=i\gamma^b=\gamma^b_L&\text{ if  }a\neq b
      \end{array}\right.
  \end{equation}
while for $b=0$ one has \begin{equation}
    \label{eq:rhoadgamma2}
\gamma^a (\gamma_L^0)^\dag \gamma^a =    \gamma^a \gamma^0 \gamma^a
      =\left\{
        \begin{array}{ll}
          \gamma^0=\gamma^0_L& \text{ for } a=0,\\
        -\gamma^0=-\gamma^0_L& \text{ for } a=1, 2, 3.\end{array}\right.
  \end{equation}
From \eqref{eq:rhoadgamma2}, the only possibility that
\eqref{eq:lorentz+adj} holds true for $b=0$ is that $a=0$. One then
checks from \eqref{eq:rhoadgamma1} that \eqref{eq:lorentz+adj} holds
true also for $b=1, 2, 3$. Hence the result.
\end{proof}
Consider now the product $\langle \cdot , R \cdot \rangle$ on
$L^2(\Man, S)$ with $R$  a single Dirac matrix~$\gamma^a$. 
 \begin{proposition}
\label{prop:lorentz} The lorentzian Dirac matrices are unitary with respect to the product above if, and only if, $R=\gamma^0$.
 \end{proposition}
  \begin{proof}
By definition of $\rho$-unitary, 
 $(\gamma_L^b)^+=\gamma^a\,\gamma_L^\dagger \gamma^a$.
Lemma \ref{lem:loentz} shows that this is equal to $\gamma_L^b$ for any $b=0, 1, 2, 3$ if and
only if $a=0$.
  \end{proof}
 
Let  $\calU_\rho(\Man, S)$  denote the group of $\rho$-unitary
operators \eqref{eq:OOounit} of $\calB(L^2(\Man ,S))$ for the twisted
product implemented by $R=\gamma^0$. It contains (the
representation~of) the group ${\cal U}_\rho$ of $\rho$-unitaries of
the algebra $\cinf\otimes\mathbb C^2$, but is bigger than it. 
  \begin{proposition}
 \label{prop:Lor}
 In the minimal twist of a $4$-dimensional, closed,
riemannian spin manifold with automorphism $\rho$ implemented by
 $R=\gamma^0$, the Lorentz group is a proper sub-group of
 $\calU_\rho(\Man, S)$.
 \end{proposition}
 \begin{proof}
By proposition \ref{prop:lorentz}, the lorentzian Dirac matrices are
$\rho$-adjoint. 
The same is true for the generators $T^{ab}$ \eqref{eq:genT}: being $+$ an involution, for any ${\cal O}, {\cal O}'$ in ${\cal
  B}(\HH)$ one has $[{\cal O}, {\cal O}']^+=-[{\cal
  O}^+, {{\cal O}'}^+]$ as well as $(i\I)^+=-i\I$, therefore
 \begin{equation*}
 (T^{ab})^+=-[\gamma^a_L,
 \gamma^b_L]^+(\frac i4\I)^+=[(\gamma^a_L)^+,
 (\gamma^b_L)]^+(-\frac i4\I)= -\frac i4[\gamma^a_L,
 \gamma^b_L]= T^{ab}.
 \end{equation*}
 Taking the exponential,and remembering that $({\cal O}^n)^+= ({\cal
 O}^+)^n$, one gets
 \begin{align}
 S[\Lambda]^+&=\sum_{n=0}^{\infty}\frac{1}{n!}\left(\left(\frac{i}{2}t_{ab}T^{ab}\right)^n\right)^+=\sum_{n=0}^{\infty}\frac{1}{n!}\left(\left(\frac{i}{2}t_{ab}T^{ab}\right)^+\right)^n,\\
&  = \sum_{n=0}^{\infty}\frac{1}{n!}\left(-\frac{i}{2}t_{ab}T^{ab}\right)^n =\exp(-\frac{i}{2}t_{ab}T^{ab})=S[\Lambda]^{-1}.
 \end{align}
\newpage
\noindent  Thus
$S[\Lambda]^+S[\Lambda]=S[\Lambda]S[\Lambda]^+=\bbbone$,
meaning that $S[\Lambda]\in\calU_\rho(\calB(\calH))$.
 
To show that the Lorentz group is a proper sub-group of
 $\calU_\rho(\Man, S)$, it is enough to exhibit one element of the
 latter which is not in the Lorentz group.
From the form  \eqref{EDirac} of the Dirac matrices, one
checks that the generators $T^{ab}$ are block diagonal, so that 
 \begin{equation*}
 S[\Lambda]=\begin{pmatrix}
 \Lambda_{\, +} & 0  \\
 0 & \Lambda_{\, -}  \\
 \end{pmatrix}
 \end{equation*}
with
 $\Lambda_{\, \pm}$ suitable sums of products of Pauli
 matrices. The point is that not all  $\rho$-unitary operators
 $U_\rho$ on $L^2(\Man, S)$ are
 block diagonal. Writing
 \begin{equation*}
 U_\rho=\begin{pmatrix}
 \alpha& \beta  \\
 \gamma & \delta
 \end{pmatrix}
 \end{equation*}
 with $\alpha, \beta, \gamma, \delta$ four $2\times 2$ complex
 matrices, one has 
 \begin{equation}
   U_\rho^+=\gamma^0 U_\rho^\dag \gamma^0 =
   \begin{pmatrix}
     \delta^\dag & \gamma^\dag \\ \beta^\dag & \alpha^\dag
   \end{pmatrix}
 \end{equation}
So $U_\rho$ is $\rho$-unitary if and only if 
\begin{equation}
  \alpha  \delta^\dagger+\beta\beta^\dag =\gamma\gamma^\dag + \delta\alpha^\dag=\bbbone_2, \qquad
  \alpha\gamma^\dag +\beta\alpha^\dagger=\gamma\delta^\dagger+ \delta\beta^\dagger=0.
\end{equation}
A first set of solutions is given by $\beta = \gamma = 0$ and
$\alpha\delta^\dag=\I$, which includes Lorentz transformations. A
second set of solutions, which are not Lorentz transformations, is
given by  $\alpha = \delta=0 $ and $\beta$, $\gamma$ in the unitary
group $U(2)$.
 \end{proof}

The action \eqref{eq:aclorentz} of the Lorentz group is the conjugate
action with respect to the involution $+$ of $\rho$-unitaries, hence
this is one of the non-entangled group actions mentioned after lemma \ref{prop:nonent}.
This means that the action of the Lorentz group is neither a
twisted-fluctuation nor a gauge transformation.
 
 \begin{remark}
 The lorentzian Dirac matrices are antiselfadjoint, except $\gamma^0$
 which is selfadjoint. Hence for $j, k=1, 2, 3$, one has 
 \begin{equation}
(T^{jk})^\dag = \frac i4 [\gamma_L^j,
 \gamma_L^k]^\dag = -\frac i4 [(\gamma_L^j)^\dag,
 (\gamma_L^k)^\dag]=-\frac i4 [(\gamma_L^j),
 (\gamma_L^k)]= T^{jk}.
\end{equation}
Therefore for $\Lambda=t_{jk}T^{jk}$, that is a spatial rotation, one has
\begin{equation}
  S[\Lambda]^\dag = \exp(-\frac i2 t_{jk}T^{jk})=S[\Lambda]^{-1}
\end{equation}
meaning that $S[\Lambda]$ is not only $\rho$-unitary but also unitary.
On the contrary the generators $T^{0j}$ are antiselfadjoint,
meaning that for boosts $\Lambda=t_{0j}T^{oj}$, the spin
representation $S[\Lambda]$ is selfadjoint, hence $\rho$-unitary but not unitary.
 \end{remark}

One may extend the action \eqref{eq:aclorentz} of the Lorentz group to the whole  of
$\cal{U}_\rho(\Man, S)$, but there is no guarantee that this leaves the fermionic
action invariant, unless one also imposes the transformation of the
real structure
\begin{equation}
\label{eq:transfoJ}
  \J\to U_\rho \,\J \,U_\rho \quad\text{ for }\; U_\rho\in{\cal U_\rho}(\Man, S)
\end{equation}
(in that case the transformation is simply a change of base
by a matrix unitary for the $\rho$-product, hence the fermionic action
is automatically conserved). For $U_\rho=S[\Lambda]$ in the  Lorentz
group, the condition \eqref{eq:transfoJ} is automatically satisfied
\cite[lemma 6.1]{martinetti2022lorentzian}.

\newpage \subsection{Spectral action}
\label{subsec:specact}

The spectral action for a usual spectral triple $(\A, \HH, D)$  is \cite{chamseddine1997spectral} 
\begin{align}
\lim_{\Lambda\to\infty}\Tr f\left(\frac{D^2}{\Lambda^2}\right)
\end{align}
where $\Lambda$ is an energy scale and $f$ a smooth approximation of
the characteristic function of the interval $[0, 1]$. This action is
invariant under the map
\begin{equation}
D\mapsto UDU^\dagger\;  \text{ with }\; U \text{ a unitary on
${\cal B}(\HH)$},
\end{equation}
since $(UDU^\dagger)^2=UD^2U^\dagger$ has the same trace
as $U^2$. Gauge transformations for usual spectral triples are of this kind.

For a minimally twisted spectral
triple, one should be careful that a gauge transformation \eqref{eq:newlook} does not
necessarily preserves the selfadjointness of $D$, so there is no
guaranty to make sense of 
$f(\text{Ad}(v) D \text{Ad}(v)^+)$ by the spectral theorem. A solution
is to work with $DD^\dagger$ instead of $D^2$, thus defining the
action as \cite{martinetti2022lorentzian} 
\begin{align}
\label{eq:specact}
\lim_{\Lambda\to\infty} 
\Tr f\left(\frac{DD^\dagger}{\Lambda^2}\right).
\end{align}
It is invariant under a twisted gauge transformation
\begin{equation}
D \mapsto
VDV^+\;  \text{ with } V:=\text{Ad}(v),
\end{equation}
since, using $(V^+)^\dag=\rho(V)$, one has
\begin{equation}
  (V D V^+)(VDV^+)^\dag=VD\rho(V)^\dagger\rho(V)D^\dagger V^\dagger =VD D^\dagger V^\dagger 
\end{equation}
has the same trace as $DD^\dagger$. But it has no
reason to be invariant under a Lorentz transformation 
\eqref{eq:aclorentz}


A Lorentz invariant action could be obtained considering $DD^+$ instead of
$DD^\dagger$. Indeed, under the map
\begin{equation}
D\mapsto U_\rho D U_\rho^+\; \text{ with }\;
U_\rho \text{ twisted unitary},
\end{equation}
one checks that 
\begin{equation}
  (U_\rho D U_\rho^+)\,(U_\rho D U_\rho^+)^+=U_\rho DD^+ U_\rho^+
\end{equation}
has the same trace as $DD^+$.  The problem is that $DD^+$ has no
reason to be selfadjoint, hence one is back to the problem of the non
selfadjointness of the argument of the function $f$ (see \cite{devastato2018lorentz} for more
on this). As a curiosity, notice that the trace of $DD^+$ is also invariant under the
map \eqref{eq:twistact} that generates the co-exact torsion, namely
\begin{equation}
  D\mapsto U_\rho D U_\rho^\dag \; \text{ with } U_\rho \text{ twisted unitary}.
\end{equation}
Indeed,  
\begin{equation}
  (U_\rho\, D \,U_\rho^\dag) \, (U_\rho D U_\rho^\dag)^+ = U_\rho\, D
  \, U^\dag_\rho\, \rho(U_\rho) D^+ U_\rho^+= U_\rho DD^+ U_\rho^+
\end{equation}
has the same trace as $DD^+$.

To conclude, we compute the action \eqref{eq:specact} for the minimal twist of a
manifold. The same action has been computed in \cite{TwistSpontBreakDevastaMartine2017}, but at the time we had
not understood that the new $1$-form field (called vector
field there) was a torsion, and its geometrical meaning was obscured by the
coupling with the degrees of freedom of the finite geometry of the
Standard Model. 

Let us consider the selfadjoint twisted covariant Dirac operator 
$\ds_{\omega_f}$ \eqref{eq:domegaf}. One has 
\begin{align*}
\ds_{\omega_f}^2&=(-i\gamma\indices{^\mu}\tilde{\nabla}_\mu^\Sp-i\gamma^\mu f_\mu \gamma)(-i\gamma\indices{^\nu}\tilde{\nabla}\indices{_\nu^\Sp}-i\gamma^\nu f_\nu \gamma)\\
&=  -\gamma\indices{^\mu}\tilde{\nabla}_\mu^\Sp\gamma\indices{^\nu}\tilde{\nabla}\indices{_\nu^\Sp}-\{\gamma\indices{^\mu}\tilde{\nabla}_\mu^\Sp, \gamma^\nu f_\nu \gamma\}+\gamma^\mu \gamma^\nu f_\mu f_\nu\\
&=\tilde{\Delta}^\Sp+\frac{1}{4}s-\{\gamma\indices{^\mu}\tilde{\nabla}_\mu^\Sp, \gamma^\nu f_\nu \gamma\}+(\frac{1}{2}[\gamma^\mu,\gamma^\nu]+g^{\mu\nu})f_\mu f_\nu\\
&=\tilde{\Delta}^\Sp+\frac{1}{4}s-\{\gamma\indices{^\mu}\tilde{\nabla}_\mu^\Sp, \gamma^\nu f_\nu \gamma\}+ f^\mu f_\mu 
\end{align*}
with
$\tilde{\Delta}^\Sp=-g^{\mu\nu}(\tilde{\nabla}_\mu^\Sp\tilde{\nabla}_\nu^\Sp-
\Gamma^{\lambda}_{\mu \nu}\tilde{\nabla}_\lambda^\Sp)$ the Laplacian
associated with the spin connection $\tilde{\nabla}\indices{^\Sp}$,
where we used Lichnerowicz formula in the third line, $s$ being the scalar curvature, see
\cite[theorem 9.16]{gracia2013elements}. The skew symmetry of $ \tilde{\Gamma}^{a}_{\mu b}$ in $a$ and $b$ implies that $\tilde{\nabla}_\mu^\Sp=\partial_\mu+ \frac{1}{4} \tilde{\Gamma}^{a}_{\mu b}\gamma_a\gamma^b $. Using  $\tilde{\nabla}_\mu^\Sp(\gamma^\nu)=c(\tilde{\nabla}_\mu dx^\lambda)=-\gamma^\lambda \Gamma^{\nu}_{\mu \lambda}$, we obtain   
\begin{align*}
\!\!\! \!\!\! \!\!-\{\gamma^\mu\tilde{\nabla}_\mu^\Sp, \gamma^\nu f_\nu \gamma\}&=-\gamma^\mu\tilde{\nabla}_\mu^\Sp \gamma^\nu f_\nu \gamma - \gamma^\nu f_\nu \gamma\gamma^\mu\tilde{\nabla}_\mu^\Sp\\
&= -\gamma^\mu(\tilde{\nabla}_\mu^\Sp \gamma^\nu) f_\nu \gamma-\gamma^\mu\gamma^\nu(\tilde{\nabla}_\mu^\Sp f_\nu) \gamma-\gamma^\mu\gamma^\nu f_\nu(\tilde{\nabla}_\mu^\Sp \gamma)-\gamma[\gamma^\mu,\gamma^\nu]f_\nu \tilde{\nabla}_\mu^\Sp\\
&=\gamma^\mu \gamma^\lambda \Gamma^{\nu}_{\mu \lambda}f_\nu \gamma -\gamma^\mu\gamma^\nu(\tilde{\nabla}_\mu^\Sp f_\nu) \gamma-\gamma^\mu\gamma^\nu f_\nu(\tilde{\nabla}_\mu^\Sp \gamma)-\gamma[\gamma^\mu,\gamma^\nu]f_\nu \tilde{\nabla}_\mu^\Sp.
\end{align*}
Following \cite[sec 11.2]{ConnMarc08b}, $\ds_{\omega_f}^2$ can be written as $\ds_{\omega_f}^2=-(g^{\mu\nu}\partial_\mu\partial_\nu+a^\lambda\partial_\lambda+b)$ with $a$ and $b$ matrix valued functions explicitly given by $a^\lambda=a^\lambda_1+a^\lambda_2$ and $b=-\frac{1}{4}s- f^\mu f_\mu+b_1+b_1^\prime+b_2$ with:
\begin{align*}
&a^\lambda_1=\gamma[\gamma^\lambda,\gamma^\nu]f_\nu &\\
&a^\lambda_2=   g^{\mu\lambda}\frac{1}{4} \tilde{\Gamma}^{a}_{\mu b}\gamma_a\gamma^b- g^{\mu\nu}\Gamma^{\lambda}_{\mu \nu}&\\
&b_1= \gamma^\mu\gamma^\nu(\tilde{\nabla}_\mu^\Sp f_\nu) \gamma &\\
&b_1^\prime=-\gamma^\mu \gamma^\lambda\Gamma^{\nu}_{\mu \lambda}f_\nu \gamma +\gamma^\mu\gamma^\nu f_\nu(\tilde{\nabla}_\mu^\Sp \gamma)+\frac{1}{4} \tilde{\Gamma}^{a}_{\mu b}f_\nu\gamma[\gamma^\mu,\gamma^\nu]\gamma_a\gamma^b &\\
&b_2= g^{\mu\nu}\frac{1}{4} \partial_\mu(\tilde{\Gamma}^{a}_{\mu b}\gamma_a\gamma^b)+g^{\mu\nu}\frac{1}{8}\tilde{\Gamma}^{a}_{\mu b}\tilde{\Gamma}^{c}_{\nu d}\gamma_a\gamma^b\gamma_c\gamma^d-g^{\mu\nu}\frac{1}{4}\Gamma^\lambda_{\mu\nu}\tilde{\Gamma}^a_{\lambda b}\gamma_a\gamma^b.&
\end{align*}
with $a_1^\lambda\partial_\lambda$, $b_1$ and $b_1^\prime$ coming from the term $-\{\gamma\indices{^\mu}\tilde{\nabla}_\mu^\Sp, \gamma^\nu f_\nu \gamma\}$, and $a_2^\lambda\partial_\lambda$ and $b_2$ coming from $\tilde{\Delta}^\Sp$.
Defining the connection
$\bar{\nabla}_\mu=\partial_\mu+\bar{\omega}_\mu$ with
$\bar{\omega}_\mu=\frac{1}{2}g_{\mu\nu}(a^\nu+g^{\lambda\rho}\Gamma^\nu_{\lambda\rho})$,
one gets 
\begin{align*}
\ds_{\omega_f}^2=-(g^{\mu\nu}\bar{\nabla}_\mu\bar{\nabla}_\nu+E)
\end{align*}
with $E=b-g^{\mu\nu}(\partial_\mu\bar{\omega}_\nu+\bar{\omega}_\mu \bar{\omega}_\nu-\bar{\omega}_\lambda \Gamma^\lambda_{\mu\nu} )=b+\tilde{b}$.

The spectral action is 
\begin{align}
\lim_{\Lambda\to \infty}\Tr \exp(-\ds_{\omega_f}^2 / \Lambda^2)\simeq \sum_{n\geq 0} \Lambda^{2m-n}a_{n}(\ds_{\omega_f}^2)
\end{align}
Following \cite[theorem 3.3.1]{gilkey2003asymptotic} (or
\cite{chamseddine2010noncommutative}), we compute the first three non
vanishing Seeley-DeWitt coefficients
\begin{align*}
a_0(\Dir_{\omega_f}^2)&=\frac{1}{(4\pi)^m}\int_\Man \Tr(\bbbone_{2^m}) dv=\frac{1}{(2\pi)^m}\int_\Man dv\\
a_2(\Dir_{\omega_f}^2)&=\frac{1}{(4\pi)^m}\int_\Man \Tr(E+\frac{s}{6}) dv=\frac{1}{(4\pi)^m}\int_\Man (-\frac{2^ms}{12}-2^mf^\mu f_\mu +b_1+b_2+\tilde{b})dv\\
a_4(\Dir_{\omega_f}^2)&=\frac{1}{360(4\pi)^m}\int_\Man \Tr(12\bar{\Delta} s+5s^2-2R_{\mu\nu}R^{\mu\nu}+2R_{\mu\nu\lambda\rho}R^{\mu\nu\lambda\rho}+60sE \\
&\qquad\qquad\qquad\qquad +60\bar{\Delta} E+180E^2+30\bar{\Omega}_{\mu\nu}\bar{\Omega}^{\mu\nu}) dv
\end{align*}
with $\bar{\Omega}_{\mu\nu}=\partial_\mu \bar{\omega}_\nu -\partial_\nu\bar{\omega}_\mu+[\bar{\omega}_\mu,\bar{\omega}_\nu]$ the field strenght of $\bar{\nabla}$ and $\bar{\Delta}\defeq \bar{\nabla}_\mu\bar{\nabla}^\mu$. The contribution of the terms $\gamma^\mu\gamma^\nu(\tilde{\nabla}_\mu^\Sp f_\nu) \gamma$ and $\gamma^\mu \gamma^\lambda\Gamma^{\nu}_{\mu \lambda}f_\nu \gamma$ in $a_2$ disappear since $\Tr(\gamma^\mu\gamma^\nu\gamma)=0$.\\

Using $\Tr(\Delta M)=\Delta\Tr(M)$ for any matrix $M$ and Laplacian
$\Delta$, together with the vanishing of the integral of the Laplacian
of a function over a closed manifold  by Stokes' theorem, the development of $a_4(\ds_{\omega_f}^2)$ gives:
\begin{align*}
&\!\!\!\!\!\!\!\!\!\!  \frac{1}{360(4\pi)^m}\int_\Man \Tr(5s^2-2R_{\mu\nu}R^{\mu\nu}+2R_{\mu\nu\lambda\rho}R^{\mu\nu\lambda\rho}-15s^2-60sf^\mu f_\mu+180(f^\mu f_\mu)^2\\
&\qquad\qquad\quad\qquad+\frac{45}{4}s^2+180\gamma^\mu\gamma^\nu\gamma^\rho\gamma^\lambda(\tilde{\nabla}_\mu^\Sp f_\nu)(\tilde{\nabla}_\rho^\Sp f_\lambda)+45sf^\mu f_\mu+30\bar{\Omega}_{\mu\nu}\bar{\Omega}^{\mu\nu}\\
&\qquad\qquad\quad\qquad+60s(b_1+b_2+\tilde{b})+180(E^2-(b_1^\prime)^2-(-\frac{1}{4}s-f^\mu f_\mu)^2)) ) dv\\
&\!\!\!\!\!\!\!\!\!\!\!\!=\frac{1}{360(4\pi)^m}\int_\Man (\frac{5}{4}s^2-2R_{\mu\nu}R^{\mu\nu}+2R_{\mu\nu\lambda\rho}R^{\mu\nu\lambda\rho}-15sf^\mu f_\mu+180(f^\mu f_\mu)^2+30\bar{\Omega}_{\mu\nu}\bar{\Omega}^{\mu\nu}\\
&\qquad\quad\qquad\qquad+180((\tilde{\nabla}_\mu^\Sp f^\mu)(\tilde{\nabla}_\nu^\Sp f^\nu)+(\tilde{\nabla}_\mu^\Sp f^\nu)(\tilde{\nabla}_\nu^\Sp f^\mu)-(\tilde{\nabla}_\mu^\Sp f_\nu)(\tilde{\nabla}^{\mu,\Sp} f^\nu))\\
&\qquad\quad\qquad\qquad+60s(b_1+b_2+\tilde{b})+180(E^2-(b_1^\prime)^2-(-\frac{1}{4}s-f^\mu f_\mu)^2)) ) dv
\end{align*}
where we use $\Tr(\gamma^\mu\gamma^\nu\gamma^\rho\gamma^\lambda)=g^{\mu\nu}g^{\rho\lambda}+g^{\mu\lambda}g^{\nu\rho}-g^{\mu\rho}g^{\nu\lambda}$. We note that many terms disappear thanks to the relations $\Tr(\gamma\gamma^\mu)=\Tr(\gamma\gamma^\mu\gamma^\nu)=\Tr(\gamma\gamma^\mu\gamma^\nu\gamma^\rho)=0$.\\

\newpage			
			\section{Conclusion and outlook}
			\label{sec conclusion}

In this paper we have answered a question initially raised in
\cite{Devastato:2013fk}, regarding the field
of $1$-forms $f_\mu dx^\mu$ obtained from the twisted fluctuation of the free part of the Dirac
 operator of the spectral triple of the Standard Model. We established
 that this field has
 a purely  geometrical interpretation: it is the Hodge dual of a $3$-form which, in case the manifold has
dimension $4$, is the lift from the tangent to the spinor bundle of an orthogonal and geodesic
preserving torsion. 
Closed $1$-form are generated by an action of the group of twisted
unitaries. Moreover the Lorentz group is a proper subgroup of these
twisted unitaries.
 
\noindent Several points will be adressed in future works:
 \begin{itemize}

\item For the minimal twist of the Standard  Model, the extra fields
  of $1$-form calculated in \cite{filaci2021minimal} should be seen as
  gauge-value torsions.  This needs to be worked out. The
 corresponding  twisted fermionic action should be computed (some preliminary
  results are in \cite{Filaci:PhD}) as well as the spectral
  action. The hope is that torsion could play
the role of the extra scalar-field introduced in
\cite{chamseddine2012resilience} to stabilise the electroweak vacuum
and fit the mass of the Higgs boson. 

\item Alternatively, one may couple a Dirac operator with torsion with
  the finite dimensional geometry of the Standard Model and study
 whether torsion couples to the Higgs so that to stabilise the
 electroweak vacuum and fit the Higgs mass. Some results in this direction have
been obtained in \cite{SpectActTorsStephan2010}.

\item Apply this framework to Weyl semimetals (see \cite{ivanov2022heat}) where Dirac operators with skew torsion have been used in  physical models.

\item In \cite{Chamseddine:2013aa}, a term which shows some similarity with our
  $1$-form field is questioned. It would be interesting to see whether an
  interpretation as torsion makes sense.

\item  More generally, one should understand the link between torsion,
  chiral asymmetry (the doubling of the algebra permits to distinguish
  between left and right components of spinors) and change of
  signature.  
\end{itemize}\smallskip

From a more mathematical side, the generalisation of the
results of this paper to arbitrary twisted spectral triples will be
studied in a forecoming work \cite{Martinetti:2024aa}.  
In particular we aim at using the group of $\rho$-unitaries as a
tool to define ``Lorentz symmetry'' for an arbitrary spectral triple, beyond the manifold case. 

Another question is the geometrical meaning of the Hodge dual of the
$1$-form  $\omega=f_\mu dx^\mu$  in dimension $n$ other than $4$, since in
that case the $n-1$-form $\star\omega$ is no longer a torsion form. In
case $\omega$ is exact, it could be interesting to make sense of the 
co-exactness of $\star\omega$ as a derivation of the Hochschild
 cycle given by the orientability axiom in Connes reconstruction theorem.


\medskip

\noindent {\bf Acknowledgments} PM thanks P. Vitale and A. Sitarz
for suggesting independently that the extra $1$-form field
was a torsion.
GN thanks L. Dabrowski and
A. Sitarz for various constructive suggestions.

 This work is part of the project \emph{Geometry of
fundamental interactions: from quantum space to quantum spacetime},
funded by the CARIGE foundation. 
PM is supported by the mathematical
physics group of INDAM, and the ``dipartimento d'eccellenza'' initiative
from the italian ministery of research and higher education.


		\appendix
		\renewcommand\thesection{\Alph{section}}
		
		\newpage
		\section{Appendices}
		\label{sec appendices}
		\addtocontents{toc}{\protect\setcounter{tocdepth}{1}}
		
		\subsection{Orthonormal frame and Hodge duality}
		\addtocontents{toc}{\protect\setcounter{tocdepth}{2}}
		\label{sec:aporto}
Let $\Man$ be a riemannian manifold of dimension $n$ with metric $g$.
The orthonormal sections of the frame bundle and its dual are
		\begin{equation}
			\label{eq:baseort}
			\left\{ E_a, \, a=1, ..., n\right\}, \quad \quad \left\{ \theta_a, \, a=1, ..., n\right\} 
			\quad \text{ such that }\quad
			\langle\theta^a, E_b\rangle =\delta^a_b.
		\end{equation}
The orthonormal
	frame coincides in any point $p$ of $\Man$ with the coordinate basis associated with the normal
	coordinates in $p$, but this is true only in $p$, not around $p$ (unless
	the Riemann tensor in $p$ vanishes). That is
	why \eqref{eq:baseort} is also called \emph{non local}
	or \emph{non-coordinate} basis.

		Given a local chart $\left\{ x^\mu \right\}$, the \emph{vielbein}
		$e^a_\mu, e_a^\mu\in C^\infty(\Man)$ are the
		coefficients of the non local basis in the local frame:
		\begin{equation}
			\label{eq:vielb}
			E_a = e_a^\mu \partial_\mu, \qquad \theta^a = e^a_\mu dx^\mu. 
		\end{equation}
		By duality one has
		\begin{equation}
			\label{eq:inviel}
			\delta^a_b = \langle \theta^a, E_b\rangle =   \langle
			e^a_\mu dx^\mu, \,e_b^\nu \partial_\nu \rangle =  e^a_\mu e_b^\mu \,\langle dx^\mu, \partial_\nu \rangle= e^a_\mu e_b^\nu \delta^\mu_\nu = e^a_\mu e_b^\mu.
		\end{equation}
                \begin{proposition}
                The	expression of  the local basis in the
                orthonormal one is given by
		\begin{equation}
			\label{eq:vielbinv}
			\partial_\mu =e_\mu^aE_a, \quad dx^\mu=e^\mu_a \theta^a.
		\end{equation}  
                \end{proposition}
                \begin{proof}
                  The inverse $\tilde e^b_\mu, \tilde e^\mu_b$ of the
                  vielbein, defined as
                  \begin{equation}
\partial_\mu ={\tilde e}_\mu^b E_b,\quad dx^\mu=\tilde e^\mu_b \theta^b
\end{equation}
satisfies (from \eqref{eq:vielb})
\begin{equation}
e^\mu_a {\tilde e}^b_\mu =\delta^b_a, \quad e^a_\mu\tilde
                  e^\mu_b =\delta^a_b.
\end{equation}
Defining  $\tilde\theta^b=\tilde e^b_\mu dx^\mu$ and $\widetilde
                  E_b= \tilde e^\mu_b \partial_\mu$, ome checks that 
                  \begin{equation}
                  \langle \theta^b-\tilde\theta^b, E_a\rangle=0,\quad \langle  \theta^a, E_b -\tilde E_b\rangle=0
                \end{equation}
for
                  any $a,b$, meaning that $\theta^b=\tilde\theta^b$
                  and  $E_b -\tilde E_b$
                  for any $b$. Therefore $\tilde  e_\mu^b=e_\mu^b$  and $\tilde
                  e^\mu_b=e^\mu_b$. Hence the result.
\end{proof}

		The components of the metric is related to the vielbein through
		\begin{align}
			\label{eq:metviel}
			g_{\mu\nu}&=g(\partial_\mu, \partial_\nu) = g(e_\mu^a E_a, e_\nu^b
			E_b)= e_\mu^a e_\nu^b \delta_{ab}.\\
	\delta_{ab}&=g(E_a, E_b) = g(e^\mu_a \partial_\mu,
                     e^\nu_b\partial_\nu)= g_{\mu\nu}\,e^\mu_a\, e^\nu_b.
				\end{align}
	
\medskip 
The Hodge
	dual of a $k$-form $\omega$ (with components $\omega_{\nu_1
		... \nu_k})$ is the $(n-k)$-form $\star \omega$ with components
	\begin{equation}
		\star \omega_{\mu_{k+1}\dots
			\mu_{n}}=\frac{\sqrt{|\text{det} \,g|}}{(n-k)!}\,\epsilon_{\mu_1\dots
			\mu_{n}}\,g^{\mu_1\nu_1}\dots g^{\mu_k\nu_k}\,\omega_{\nu_1\dots \nu_k}.
	\end{equation}
	
	In the non-local orthonormal frame, the formula simplifies as
	\begin{equation}
		\label{eq:Hodge}
		\star \omega_{a_{k+1}\dots  a_{n}}=\frac{1}{(n-k)!}\,\epsilon_{a_1\dots
			a_{n}}\,\delta^{a_1 b_1}\dots \delta^{a_k b_k}\,\omega_{b_1\dots b_k}.
	\end{equation}

		\addtocontents{toc}{\protect\setcounter{tocdepth}{1}}
		
		\addtocontents{toc}{\protect\setcounter{tocdepth}{2}}
		
		\addtocontents{toc}{\protect\setcounter{tocdepth}{1}}
\newpage
		\subsection{Dirac matrices and Clifford action}
		\label{ConvGamMat}
		\addtocontents{toc}{\protect\setcounter{tocdepth}{2}}

		Let $\sigma_{j= 1,2,3}$ be the Pauli matrices:
		\begin{equation}
			\label{Pauli}
			\sigma_1 = \left(\begin{array}{cc} 0 & 1 \\ 1 & 0 \end{array}\right)\!,
			\qquad	\sigma_2 = \left(\begin{array}{cc} 0 & -i \\ i & 0 \end{array}\right)\!,
			\qquad	\sigma_3 = \left(\begin{array}{cc} 1 & 0 \\ 0 & -1 \end{array}\right)\!.
		\end{equation}%
		In four-dimensional euclidean space, the Dirac matrices (in chiral
		representation) are
		\begin{equation}
			\label{EDirac}
			\gamma^a =
			\left( \begin{array}{cc}
				0 & \sigma^a \\ \tilde\sigma^a & 0
			\end{array} \right)\!, \qquad
			\gamma \coloneqq  \gamma^1\,\gamma^2\,\gamma^3\,\gamma^0 =
			\left( \begin{array}{cc}
				\mathbb{\bbbone}_2 & 0 \\ 0 & -\mathbb{\bbbone}_2 
			\end{array} \right)\!,
		\end{equation}
		where, for  $a= 0,j$, we define
		\begin{equation}
			\label{eq:defsigmamu}
			\sigma^a \coloneqq  \left\{ \mathbb{\bbbone}_2, -i\sigma_j \right\}\!, \qquad
			\tilde\sigma^a \coloneqq  (\sigma^a)^\dag\left\{ \mathbb{\bbbone}_2, i\sigma_j \right\}\!.
		\end{equation}
		They satisfy the anticommutation relation
		\begin{equation}
			\gamma^a\gamma^b + \gamma^b \gamma^a =2\delta^{ab}\mathbb \bbbone_4 \quad
			\forall a, b= 0, ..., 3.
		\end{equation}
		
		On a riemannian spin manifold of dimension $4$, the Dirac
		matrices are linear combinations
		\begin{equation}
			\gamma^\mu =
			e^\mu_a \gamma^a
			\label{eq:111}
		\end{equation}
		of the euclidean ones, where $\left\{e_\mu^a\right\}$ are the
		vierbein defined by the metric. 
	 	They are selfadjoint, unitary, matrices such that
		\begin{equation}
\label{eq:anticommute}
			\{\gamma^\mu,\gamma^\nu\}=2g^{\mu\nu}\mathbb \bbbone_{2^m}.
		\end{equation}
		This is the index which tells whether we are considering the euclidean
		matrices \eqref{EDirac} (latin index) or the riemannian ones
		\eqref{eq:111} (greek index).
		\smallskip

		These definitions extends to any 
		manifold of even dimension $n=2m$.  We still denote
                $\gamma^a$ the set of $n$ square matrices of dimension $2^m$
                satisfying  \eqref{eq:anticommute} and denote
		\begin{equation}
			\label{eq:gamma5}
			\gamma=-(-i)^m\prod_{a=0}^{2m-1}\gamma^a
			\end{equation}
			the analogue of $\gamma$ in dimension $n$. One has
			\begin{equation}
				\label{eq:gammpermut}
				\gamma =-\frac{(-i)^m}{(2m)!}\epsilon_{a_1\dots a_{2m}}\gamma^{a_1}\dots\gamma^{a_{2m}}
			\end{equation}
			where the Levi-Cevita symbol
			\begin{equation}
				\label{eq:symlv}
				\epsilon_{a_0\dots a_{2m-1}} \text{ is }\left\{
				\begin{array}{ll}
					0 &\text{if at least two indices $a_k, a_l$ are equal,} \\[4pt]
					(-1)^p &\text{when all the indices are different, with $p$ the sign }\\[4pt]
					& \text{of
						the permutation $a_0 a_1...a_{2m-1}\longleftrightarrow 01...2m-1$.}
				\end{array}\right.
			\end{equation}
			In particular one has $\epsilon_{012\dots
                          2m-1}=1$.
\medskip

The
	Clifford action of a $p$-form 
        \begin{equation}
          \omega_p= \omega_{\mu_1 .. \mu_p} dx^{\mu_1}\wedge
		... \wedge dx^\mu_p 
        \end{equation}
is
	\begin{equation}
		\label{eq:Cliffaction}
	  c(\omega_p) := \omega_{\mu_1 .. \mu_p} \gamma^{\mu_1} ... \gamma^{\mu_p}.
	\end{equation}

			

			\addtocontents{toc}{\protect\setcounter{tocdepth}{1}}
			\subsection{Real structure and grading}
\label{app:relagrad}
			\addtocontents{toc}{\protect\setcounter{tocdepth}{2}}

We list several useful results on the minimal twist of an even
dimensional manifold described in
			\S~\ref{subsec:miniamltwist}.

\begin{proposition} On an even dimensional riemannian manifold of
  dimension $2m$, for
  any $\mu=1, ..., m$ and $a\in\cinf\otimes\C^2$, one has
  \begin{align}
\label{eq:kodim4}&\gamma^\mu a = \rho(a)\gamma^\mu,\\
\label{eq:kodim5}
&\J\gamma^\mu = -\gamma^\mu \J ,\\
\label{eq:kodim40}
&\J a\J^{-1}= \left\{
    \begin{array}{ll}
      a^*&\text{ in  $KO$ dimension }  0,4,\\
    \rho(a^*) &\text{ in $KO$ dimension } 2, 6.\end{array}
\right.
\end{align}
\end{proposition}
\begin{proof}
  The representation \eqref{FormOfa} is
  \begin{equation}
    \pi(a)  =\frac{\bbbone + \gamma}2 \pi_0(f) +  \frac{\bbbone - \gamma}2 \pi_0(g)
    \quad\quad \forall a=(f,g)\in C^\infty(\Man)\otimes \mathbb C^2
  \end{equation}
where $\pi_0$ is the usual representation by multiplication of $\cinf$
on $L^2(\Man, S)$ used in the spectral triple \eqref{eq:tscano}.
By definition of the grading $\gamma$ (which
  has constant coefficients) anti-commutes with $\ds$, hence any Dirac
  matrix anti-commutes with $\gamma$. Therefore
  \begin{equation}
    \label{eq:twistcomgpi}
    \gamma^\mu\pi(a) = \left(\frac{\bbbone - \gamma}2 \pi_0(f) +
      \frac{\bbbone + \gamma}2 \pi_0(g)\right) \gamma^\mu=
    \pi(\rho(a))\,\gamma^\mu \quad\forall \mu=1, ..., 2m.
  \end{equation}

Similarly, in $KO$ dimension $2, 6$ the real structure anticommutes
with the grading, while in $KO$ dimension $0,4$ the two operators
commute, hence
 \begin{align}
    \label{eq:twistcomgpiJ26}
    \J\pi(a)\J^{-1} = \left(\frac{\bbbone - \gamma}2 \J\pi_0(f)J^{-1} +
      \frac{\bbbone + \gamma}2 \J\pi_0(g)\J^{-1}\right) =
   \pi(\rho(a^*))&\text{in $KO$-dim. $2,6$},\\
     \label{eq:twistcomgpiJ04}
    \J\pi(a)\J^{-1} = \left(\frac{\bbbone + \gamma}2 J\pi_0(f)\J^{-1} +
      \frac{\bbbone - \gamma}2 \J\pi_0(g)\J^{-1}\right) = \pi(a^*) &\text{in $KO$-dim. $2,6$}
  \end{align}
 where we use $\J\pi_0(f) \J^{-1}= \pi_0(f^*)$ for any $f\in\cinf$.
\end{proof}
			Note that this proof is independent of the chart, and does not require
			the explicit form of the Dirac matrices.

			\newpage
			

		\end{document}